\documentclass[final]{elsarticle}
\usepackage{amssymb, amsmath, amsthm}
\usepackage{latexsym}
\usepackage{textcomp}
\usepackage{epsfig}
\usepackage{subfigure}
\usepackage{times}
\usepackage{xspace}
\usepackage{multirow}
\usepackage{rotating}
\usepackage{calc}
\usepackage{wrapfig}
\usepackage{graphicx}
\usepackage[inline]{enumitem}
\usepackage{lineno}

\newtheorem{theorem}{Theorem} 
\newtheorem{lemma}{Lemma}
\newtheorem{property}{Property}
\newtheorem{corollary}{Corollary}

\graphicspath{{figures/}}

\newcommand{\NP}{\ensuremath{\mathrm{NP}}\xspace}
\newcommand{\Left}{\lambda}
\newcommand{\Right}{\rho}
\newcommand{\FPP}{\texttt{dFPP}\xspace}
\newcommand{\EFPP}{\texttt{RAC-drawer}\xspace}
\newcommand{\U}{\mathcal{U}}
\newcommand{\A}{\mathcal{A}}
\newcommand{\true}{\texttt{true}\xspace}
\newcommand{\false}{\texttt{false}\xspace}
\newcommand{\IN}{\texttt{IN}\xspace}
\newcommand{\OUT}{\texttt{OUT}\xspace}
\newcommand{\Rac}{\textsc{RAC}\xspace}

\DeclareMathOperator{\ktop}{top}

\newcommand{\remove}[1]{}

\begin{document}

\begin{frontmatter}

\title{Recognizing and Drawing IC-planar Graphs\tnoteref{t1}}
\tnotetext[t1]{The research is supported in part by the Deutsche
  Forschungsgemeinschaft (DFG), grant Br835/18-1; the MIUR project
  AMANDA ``Algorithmics for MAssive and Networked DAta'',
  prot. 2012C4E3KT\_001; the ESF EuroGIGA project GraDR (DFG grant
  Wo~758/5-1); and NSERC Canada Discovery Grant. An extended abstract
  of this paper has been presented at the $23^{rd}$ International
  Symposium on Graph Drawing, GD 2015~\cite{BrandenburgDEKL15} (see
  also~\cite{bdeklm-rdicg-arXiv15} for a technical
  report.)} 

\author[inst1]{Franz J. Brandenburg}
\ead{brandenb@fim.uni-passau.de}
\author[inst2]{Walter Didimo}
\ead{walter.didimo@unipg.it}
\author[inst3]{William S. Evans}
\ead{will@cs.ubc.ca}
\author[inst4]{Philipp Kindermann\corref{cor1}}
\ead{philipp.kindermann@fernuni-hagen.de}
\author[inst2]{Giuseppe Liotta}
\ead{giuseppe.liotta@unipg.it}
\author[inst2]{Fabrizio Montecchiani\corref{cor1}}
\ead{fabrizio.montecchiani@unipg.it}

\cortext[cor1]{Corresponding authors} 

\address[inst1]{Universit\"at Passau, Germany}
\address[inst2]{Dipartimento di Ingegneria, Universit\`a degli Studi di Perugia, Italy}
\address[inst3]{University of British Columbia, Canada}
\address[inst4]{FernUniversit\"at in Hagen, Germany}

\begin{abstract}
We give new results about the relationship between \emph{1-planar graphs} and \emph{\Rac graphs}. A graph is 1-planar if it has a drawing where each edge is crossed at most once. A graph is \Rac if it can be drawn in such a way that its edges cross only at right angles. These two classes of graphs and their relationships have been widely investigated in the last years, due to their relevance in application domains where computing readable graph layouts is important to analyze or design relational data sets. 
We study \emph{IC-planar graphs}, the sub-family of 1-planar graphs that admit 1-planar drawings with \emph{independent crossings} (i.e., no two crossed edges share an endpoint). We prove that every IC-planar graph admits a straight-line RAC drawing, which may require however exponential area. If we do not require right angle crossings, we can draw every IC-planar graph with straight-line edges in linear time and quadratic area. We then study the problem of testing whether a graph is IC-planar. We prove that this problem is NP-hard, even if a rotation system for the graph is fixed. On the positive side, we describe a polynomial-time algorithm that tests whether a triangulated plane graph augmented with a given set of edges that form a matching is IC-planar.
\end{abstract}


\end{frontmatter}

\section{Introduction}\label{se:introduction}
\emph{Graph drawing} is a well-established research area that studies how to automatically compute visual representations of relational data sets in many application domains, including software engineering, circuit design, computer networks, database design, social sciences, and biology (see, e.g.,~\cite{dett-gd-99,dl-gvdm-07,jm-gds-03,kw-dg-01,s-gd-02,t-hgd-13}). The aim of a graph visualization is to clearly convey the structure of the data and their relationships, in order to support users in their analysis tasks. In this respect, there is a general consensus that edges with many crossings and bends negatively affect the readability of a drawing of a graph, as also witnessed by several user studies on the subject (see, e.g.,~\cite{DBLP:journals/iwc/Purchase00,DBLP:journals/ese/PurchaseCA02,DBLP:journals/ivs/WarePCM02}). At the same time, more recent cognitive experiments suggest that edge crossings do not inhibit user task performance if the edges cross at large angles~\cite{DBLP:conf/apvis/Huang07,DBLP:journals/vlc/HuangEH14,DBLP:conf/apvis/HuangHE08}. As observed in~\cite{del-dgrac-2011}, intuitions of this fact can be found in real-life applications: for example, in hand-drawn metro maps and circuit schematics, where edge crossings are conventionally at 90 degrees (see, e.g.,~\cite{Vignelli}), and in the guidelines of the CCITT (Comit\'e Consultatif International T\'el\'ephonique et T\'el\'egraphique) for drawing Petri nets, where it is written: ``There should be no acute angles where arcs cross''~\cite{CCITT-85}.

The above practical considerations naturally motivate the theoretical study of families of graphs that can be drawn with straight-line edges, few crossings per edge, and right angle crossings at the same time. This kind of research falls in an emerging topic of graph drawing and graph theory, informally called ``beyond planarity''. The general framework of this topic is to relax the planarity constraint by allowing edge crossings, but still forbidding those configurations that would affect the readability of the drawing too much. Different types of forbidden edge-crossing configurations give rise to different families of beyond planar graphs. For example, for any integer $k \geq 3$, the family of \emph{$k$-quasi planar graphs} is the set of graphs that have a drawing with no $k$ mutually crossing edges (see, e.g.,~\cite{DBLP:journals/jct/AckermanT07,aapps-qpgln-C97,DBLP:journals/siamdm/FoxPS13}). For any positive integer $k$, the family of \emph{$k$-planar graphs} is the set of graphs that admit a drawing with at most $k$ crossings per edge~\cite{pt-gdfce-C97}; in particular, \emph{$1$-planar graphs} have been widely studied in the literature (see, e.g.,~\cite{abk-slgd3-GD13,DBLP:journals/ipl/Didimo13,DBLP:journals/algorithmica/GrigorievB07,help-ft1pg-COCOON12,DBLP:journals/jgt/KorzhikM13,r-esadk-AMS65}). \emph{\Rac (Right Angle Crossing) graphs} are those graphs that admit a drawing where edges cross only at right angles (see, e.g.,~\cite{del-dgrac-2011}). Several algorithms and systems for computing \Rac drawings or, more in general, large angle crossing drawings, have been described in the literature (see, e.g.,~\cite{DBLP:journals/cj/ArgyriouBS13,DBLP:journals/algorithmica/GiacomoDEL14,DBLP:journals/cj/DiGiacomoDGLR15,DBLP:journals/comgeo/GiacomoDLM13,dlr-tdfda-10,DBLP:journals/jgaa/DidimoLR11,nehh-lcacl-2010,DBLP:conf/vl/HuangEHL10}). See also~\cite{dl-cargd-12} for a survey. 

In this scenario, special attention is receiving the study of the relationships between 1-planar drawings and RAC drawings with straight-line edges. The maximum number of edges of an $n$-vertex 1-planar graph is $4n-8$~\cite{pt-gdfce-C97}, while $n$-vertex straight-line 1-planar drawings and $n$-vertex straight-line RAC drawings have at most $4n-9$ edges~\cite{DBLP:journals/ipl/Didimo13} and $4n-10$ edges~\cite{del-dgrac-2011}, respectively. This implies that not all 1-planar graphs admit a straight-line drawing and not all 1-planar graphs with a straight-line drawing admit a straight-line RAC drawing. The characterization of the 1-planar graphs that can be drawn with straight-line edges was given by Thomassen in 1988~\cite{t-rdg-JGT88}.   

\smallskip\noindent{\bf Our Contribution.}
In this paper we give new results on the relationship between 1-planar graphs, \Rac graphs, and straight-line drawings. We concentrate on a subfamily of 1-planar graphs known as \emph{IC-planar graphs}, which stands for graphs with \emph{independent crossings}~\cite{a-cnirc-AMC08}. An IC-planar graph is a graph that admits a 1-planar drawing where no two crossed edges share an endpoint, i.e., all crossing edges form a matching. IC-planar graphs have been mainly studied both in terms of their structure and in terms of their applications to coloring problems~\cite{a-cnirc-AMC08,ks-cpgic-JGT10,z-dcmgp-AMS14,zl-spgic-CEJM13}. We prove that:

\begin{itemize}
\item Every IC-planar graph with $n$ vertices admits a (non-RAC) straight-line drawing in~$O(n^2)$ area, which can be computed in $O(n)$ time (Theorem~\ref{th:straightline}). Our bound on the area requirement is worst-case optimal. We recall that there are embedded 1-planar graphs whose straight-line drawings require~$\Omega(2^n)$ area~\cite{help-ft1pg-COCOON12}.

\item Every IC-planar graph is a \Rac graph (Theorem~\ref{th:rac-drawings}), but we also show that a straight-line \Rac drawing of an $n$-vertex IC-plane graph may require~$\Omega(q^n)$ area, for a suitable constant~$q > 1$ (Theorem~\ref{th:rac-area}).  
\end{itemize}

Moreover, as a natural problem related to the results above, we study the computational complexity of recognizing IC-planar graphs. Namely:

\begin{itemize}
\item We prove that IC-planarity testing is \NP-complete both in the variable embedding setting (Theorem~\ref{th:np-hard}) and when the rotation system of the graph is fixed (Theorem~\ref{th:np-rot}). Note that, 1-planarity testing is already known to be \NP-complete in general~\cite{DBLP:journals/algorithmica/GrigorievB07,DBLP:journals/jgt/KorzhikM13} and for a fixed rotation system~\cite{JGAA-347}. Testing 1-planarity remains NP-hard even for \emph{near-planar} graphs, i.e., graphs that can be obtained by a planar graph by adding an edge~\cite{DBLP:journals/siamcomp/CabelloM13}.   

\item On the positive side, we give a polynomial-time algorithm that tests whether a triangulated plane graph augmented with a given set of edges that form a matching is IC-planar (Theorem~\ref{th:triang-test}).
The interest in this special case is also motivated by the fact that every $n$-vertex IC-planar graph with maximum number of edges (i.e., with $13n/4-6$ edges) is the union of a triangulated planar graph and of a set of edges that form a matching~\cite{zl-spgic-CEJM13}. 
\end{itemize}

We finally recall that the problem of recognizing 1-planar graphs has been studied also in terms of parameterized complexity. Namely, Bannister, Cabello, and Eppstein describe fixed-parameter tractable algorithms with respect to different parameters: vertex cover number, tree-depth, and cyclomatic number~\cite{DBLP:conf/wads/BannisterCE13}. They also show that the problem remains NP-complete for graphs of bounded bandwidth, pathwidth, and treewidth, which makes unlikely to find parameter tractable algorithms with respect to these parameters. Fixed-parameter tractable algorithms have been also described for computing the crossing number of a graph, a problem that is partially related to the research on beyond planarity (see, e.g.,~\cite{DBLP:conf/stoc/Grohe01,DBLP:conf/gd/PelsmajerSS07a}).

The remainder of the paper is organized as follows. 
In Section~\ref{se:preliminaries} we recall basic definitions used in the paper. Section~\ref{ic:sec:drawing} is devoted to prove Theorem~\ref{th:straightline}. In Section~\ref{se:rac} we prove Theorems~\ref{th:rac-drawings} and~\ref{th:rac-area}.  Section~\ref{se:recognition} proves Theorems~\ref{th:np-hard},~\ref{th:np-rot}, and~\ref{th:triang-test}. Conclusions and open problems are in Section~\ref{se:conclusions}.

\section{Preliminaries}\label{se:preliminaries}

\begin{figure}[t]
    \centering
    \subfigure[]{\includegraphics[scale=0.6]{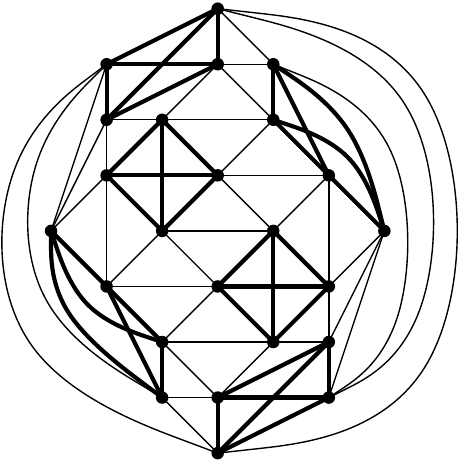}\label{fi:ic-drawing}}
    \hfil
    \subfigure[]{\includegraphics[page=1, scale=0.8]{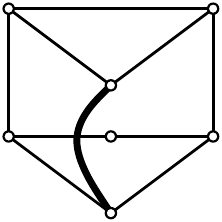}
      \includegraphics[page=2, scale=0.8]{rotsys-diffemd-ex}\label{fi:rotemb}}
    \hfil
    \subfigure[]{\includegraphics[scale=0.8,page=3]{thomassen}\label{fi:thomassen-1}}
    \hfil
    \subfigure[]{\includegraphics[scale=0.8]{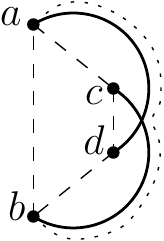}\label{fi:thomassen-2}}
  \caption{(a) An IC-planar drawing. (b) Two different IC-planar embeddings of 
    the same graph with the same rotation system. (c) An X-configuration. (d) 
    A B-configuration.}
\end{figure}

We consider simple undirected graphs~$G$. A \emph{drawing}~$\Gamma$ 
of~$G$ maps the vertices of~$G$ to distinct points in the plane and the edges 
of~$G$ to simple Jordan curves between their endpoints. If the vertices are 
drawn at integer coordinates,~$\Gamma$ is a \emph{grid drawing}.
$\Gamma$ is \emph{planar} if no edges cross, and \emph{1-planar} if each edge
is crossed at most once. $\Gamma$ is \emph{IC-planar} if it is 1-planar and 
there are no crossing edges that share a vertex. An example of an IC-planar graph is shown in Figure~\ref{fi:ic-drawing}. 

A planar drawing~$\Gamma$ of a graph~$G$ induces an \emph{embedding}, 
which is the class of topologically equivalent drawings. In particular, an 
embedding specifies the regions of the plane, called \emph{faces}, whose boundary 
consists of a cyclic sequence of edges. The unbounded face is called the 
\emph{outer face}. For a 1-planar drawing, we can still derive an embedding 
by allowing the boundary of a face to consist also of edge segments from a 
vertex to a crossing point. 
A graph with a given planar (1-planar, IC-planar) embedding is called a 
\emph{plane} (\emph{1-plane, IC-plane}) graph. 
A \emph{rotation system}~$\mathcal{R}(G)$ of a graph~$G$ describes a possible 
cyclic ordering of the edges around the vertices. $\mathcal{R}(G)$ is planar 
(1-planar, IC-planar) if~$G$ admits a planar (1-planar, IC-planar) 
embedding that preserves~$\mathcal{R}(G)$. Observe that~$\mathcal{R}(G)$ can 
directly be retrieved from a drawing or an embedding. The converse does not 
necessarily hold, as shown in Figure~\ref{fi:rotemb}.

A \emph{kite}~$K$ is a graph isomorphic to~$K_4$ with an embedding such that all
the vertices are on the boundary of the outer face, the four edges on the 
boundary are planar, and the remaining two edges cross each other; see Figure~\ref{fi:thomassen-1}. 
Thomassen~\cite{t-rdg-JGT88} characterized the possible crossing configurations 
that occur in a 1-planar drawing. Applying this characterization to IC-planar
drawings gives rise to the following property:
\begin{property}\label{pr:char-crossins}
  Every crossing of an IC-planar drawing is either an X- or a B-crossing.
\end{property}
\noindent
Here, an X-crossing has the crossing ``inside'' the 4-cycle
(see Figure~\ref{fi:thomassen-1}),
and a B-crossing has the crossing ``outside'' the 4-cycle
(see Figure~\ref{fi:thomassen-2}). 
We remark that, according to Thomassen~\cite{t-rdg-JGT88}, a
$1$-planar drawing may contain crossings that are neither X- nor
B-crossings, but W-crossings. This third type of crossing is not
possible in an IC-planar drawing since it contains two vertices
incident to two crossed edges.

Let~$G$ be a plane (1-plane, IC-plane) graph. $G$ is \emph{maximal} if no edge 
can be added without violating planarity (1-planarity, IC-planarity). A planar
(1-planar, IC-planar) graph~$G$ is maximal if every planar (1-planar, IC-planar)
embedding is maximal. 
If we restrict to 1-plane (IC-plane) graphs, we say that~$G$ is 
\emph{plane-maximal} if no edge can be added without creating at least an edge
crossing on the newly added edge (or making the graph not simple).  We call the 
operation of adding edges to~$G$ until it becomes plane-maximal a 
\emph{plane-maximal augmentation}.

\section{Straight-line drawings of IC-planar graphs}\label{ic:sec:drawing}

We show that every IC-planar graph admits an IC-planar straight-line grid 
drawing in quadratic area, and this area is worst-case optimal 
(Theorem~\ref{th:straightline}). The result is based on first using a new technique
that augments an embedding of the input graph to a maximal IC-plane graph (the
resulting embedding might be different from the original one) with specific 
properties (Lemma~\ref{lem:3-connected}), and then suitably applying a drawing
algorithm by Alam {\em et al.} for triconnected 1-plane 
graphs~\cite{abk-slgd3-GD13} on the augmented graph.
We say that a kite $(a,b,c,d)$ with crossing edges $(a,d)$ and $(b,c)$ is 
\emph{empty} if it contains no other 
vertices, that is, the edges $(a,c)$, $(a,d)$, and $(a,b)$ are consecutive in 
the counterclockwise order around~$a$; see Figure~\ref{fi:maximal-planar-augment-2}. The condition for the edges around~$b$, $c$, and~$d$ is analogous. We are now ready to prove the next lemma. 

\newcommand{\lemThreeConText}[1]{
	Let~$G=(V,E)$ be an IC-plane graph with~$n$ vertices. There exists an 
  $O(n)$-time algorithm that computes a plane-maximal IC-plane 
  graph~$G^+ = (V, E^+)$ with~$E \subseteq E^+$ such that the following
  conditions hold: 
	\begin{enumerate}[label={\bfseries (c\arabic*)}]
		\item \label{#1-kite} The four endpoints of each pair of crossing edges 
      induce a kite. 
    \item \label{#1-empty} Each kite is empty.
		\item \label{#1-triangulated} Let~$C$ be the set of crossing edges 
			in~$G^+$. Let~$C^* \subset C$ be a subset containing exactly one edge for
      each pair of crossing edges. Then~$G^+ \setminus C^*$ is plane and 
      triangulated.
    \item \label{#1-3cycle} The outer face of~$G^+$ is a $3$-cycle of non-crossed edges.
	\end{enumerate}
}
\begin{lemma} \label{lem:3-connected}
	\lemThreeConText{main}
\end{lemma} 

\begin{figure}[t]
    \centering
    \subfigure[The kite (drawn bold) is not empty]{\includegraphics[page=1]{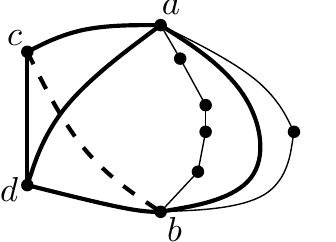}\label{fi:maximal-planar-augment-1}}\hfil
    \subfigure[Rerouting edge $(a,b)$ to make the kite empty]{\includegraphics[page=2]{maximal-planar-augment}\label{fi:maximal-planar-augment-2}}\hfil
    \subfigure[Triangulating the remaining faces]{\includegraphics[page=3]{maximal-planar-augment}\label{fi:maximal-planar-augment-3}}
  \caption{Illustration for the proof of Lemma~\ref{lem:3-connected}.}
  \label{fi:maximal-planar-augment}
\end{figure}

\begin{proof} 
  Let~$G$ be an IC-plane graph; we augment $G$ by adding edges such that for each pair of 
  crossing edges~$(a,d)$ and~$(b,c)$ the subgraph induced by vertices $\{a,b,c,d\}$ is 
  isomorphic to~$K_4$; see the dashed edges in Figures~\ref{fi:thomassen-1} and~\ref{fi:thomassen-2}. 
  Next, we want to make sure that this subgraph forms an X-configuration and 
  the resulting kite is empty.
  Since $G$ is IC-planar, it has no two B-configurations sharing an edge.
  Thus, we remove a B-configuration with vertices $\{a,b,c,d\}$ by rerouting
  the edge~$(a,b)$ to follow the edge~$(a,d)$ from vertex~$a$ until the
  crossing point, then edge~$(b,c)$ until vertex~$b$, as shown by the dotted
  edge in Figure~\ref{fi:thomassen-2}. This is always possible, because
  edges~$(a,c)$ and~$(b,d)$ only cross each other; hence, following their curves,
  we do not introduce any new crossing. The resulting IC-plane graph 
  satisfies~\ref{main-kite} (recall that, by Property~\ref{pr:char-crossins},
  only X- and B-configurations are possible).    
  Now, assume that a kite $(a,b,c,d)$ is not empty; see 
  Figure~\ref{fi:maximal-planar-augment-1}. Following the same argument as above, 
  we can reroute the edges $(a,b)$, $(b,d)$, $(c,d)$ and $(a,d)$ to follow the
  crossing edges $(a,d)$ and $(b,c)$; see Figure~\ref{fi:maximal-planar-augment-2}.
  The resulting IC-plane graph is denoted by~$G'$ and satisfies~\ref{main-empty}.
  
  We now augment~$G'$ to~$G^+$, such that~\ref{main-triangulated} is satisfied.
  Let~$C$ be the set of all pairs of crossing edges in~$G'$. 
  Let~$C^*$ be a subset constructed from~$C$ by keeping only one (arbitrary) 
  edge for each pair of crossing edges. The graph~$G' \setminus C^*$
  is clearly plane. To ensure~\ref{main-triangulated},
  graph~$G^+ \setminus C^*$ must be plane and triangulated. Because~$G'$ 
  satisfies~\ref{main-empty}, each removed edge spans two triangular faces 
  in~$G' \setminus C^*$. Thus, no face incident to a crossing edge has to be
  triangulated. We internally triangulate the other faces by picking any vertex 
  on its boundary and connecting it to all other vertices (avoiding multiple
  edges) of the boundary; see e.g. Figure~\ref{fi:maximal-planar-augment-3}. 
  Graph $G^+$ is then obtained by reinserting the edges in $C^*$ and 
  satisfies~\ref{main-triangulated}. To satisfy~\ref{main-3cycle}, notice 
  that~$G^+$ is IC-plane, hence, it has a face~$f$ whose boundary contains only 
  non-crossed edges. Also, $f$ is a $3$-cycle by construction. Thus, we can 
  re-embed~$G^+$ such that~$f$ is the outer face. 
  
  It remains to prove that the 
  described algorithm runs in~$O(n)$ time. Let~$m$ be the number of 
  edges of~$G$. Augmenting the graph such that for each pair of crossing edges 
  their endpoints induce a subgraph isomorphic to~$K_4$ can be done in~$O(m)$
  time (the number of added edges is~$O(n)$). Similarly, rerouting some edges 
  to remove all B-configurations requires~$O(m)$ time. Also, triangulating 
  the graph $G' \setminus C^*$ can be done in time proportional to the number of 
  faces of $G' \setminus C^*$, which is~$O(n+m)$. Since IC-planar graphs are 
  sparse~\cite{zl-spgic-CEJM13}, the time complexity follows.
\end{proof}

\newcommand{\thStraightLine}{
	There is an $O(n)$-time algorithm that takes an IC-plane 
	graph~$G$ with~$n$ vertices as input and constructs an IC-planar 
	straight-line grid drawing of~$G$ in~$O(n) \times O(n)$ area. 
	This area is worst-case optimal.
}

\begin{theorem}\label{th:straightline}
	\thStraightLine
\end{theorem}
\begin{proof}
  Augment~$G$ into a plane-maximal IC-plane graph~$G^+$ in~$O(n)$ time using
  Lemma~\ref{lem:3-connected}. Graph $G^+$  is triconnected, as it contains a
  triangulated plane subgraph. Draw~$G^+$ with the algorithm by Alam 
  {\em et al.}~\cite{abk-slgd3-GD13} which takes as input a 1-plane
  triconnected graph with~$n$ vertices and computes a 1-planar drawing on the
  $(2n -2 )\times(2n-3)$ grid in~$O(n)$ time; this drawing is straight-line, but for 
  the outer face, which may contain a bent edge if it has two crossing 
  edges. Since by Lemma~\ref{lem:3-connected} the outer face of~$G^+$ has no
  crossed edges,~$\Gamma$ is straight-line and IC-planar.  The edges added during the augmentation are removed from~$\Gamma$.
  
  It remains to prove that the area bound of the algorithm is worst-case 
  optimal. To this aim, we show that for every $n \geq 2$ there exists an 
  IC-planar graph~$G$ with $\Theta(n)$ vertices, such that every   IC-planar 
  straight-line grid drawing of~$G$ requires~$\Omega(n^{2})$ area.
  Dolev {\em et al.}~\cite{dlt-pepg-ACR84}
	described an infinite family of planar graphs, called nested triangle 
	graphs, such that every planar straight-line drawing of an $n$-vertex 
  graph~$G$ (for $n \geq 6$) of this family requires~$\Omega(n^2)$ area. 
  We augment~$G$ as follows. For every edge~$(u,v)$ of~$G$, we add a 
  vertex~$c_{uv}$, and two edges~$(u,c_{uv})$ and~$(c_{uv},v)$. Denote by~$G^+$
  the resulting augmented graph, which clearly has~$\Theta(n)$ vertices. 
  We now show that in every possible IC-planar 
  straight-line drawing of~$G^+$ there are no two edges of~$G$ that cross each 
  other.
  Observe that the subgraph induced by the vertices~$u,v,c_{uv}$ is a
  3-cycle.
  Any drawing of two cycles must cross an even number of times.
  If two edges~$(u,v)$ and~$(w,z)$ of~$G$
  cross each other in an IC-planar drawing of~$G^+$, then the cycles
  $u,v,c_{uv}$ and $w,z,c_{wz}$ must cross at least twice.
  Since these are 3-cycles, either some edge is crossed at least twice
  or two adjacent edges are crossed.  In either case, this violates
  the IC-planar condition.
  Hence, the 
  subgraph~$G$ must be drawn planar and this implies that any
  straight-line IC-planar drawing of $G^+$ requires~$\Omega(n^2)$ area.
\end{proof}

\section{IC-planarity and \Rac graphs}\label{se:rac}

It is known that every $n$-vertex maximally dense \Rac graph (i.e., \Rac graph with $4n-10$ edges) is 1-planar, and that there 
exist both 1-planar graphs that are not \Rac and \Rac graphs that are not 
1-planar~\cite{el-rac1p-DAM13}.  
Here, we further investigate the intersection between the classes of 
1-planar and \Rac graphs, showing that all IC-planar graphs are \Rac. To 
this aim, we describe a polynomial-time constructive algorithm. The computed 
drawings may require exponential area, which is however worst-case optimal; 
indeed, we exhibit IC-planar graphs that require exponential area in any 
possible IC-planar straight-line \Rac drawing.
Our construction extends the linear-time algorithm by de 
Fraysseix {\em et al.}~\cite{dpp-hdpgg-C90} that computes a planar straight-line grid drawing of a 
maximal (i.e., triangulated) plane graph in quadratic area;
we call it the~\FPP algorithm. We need to recall the 
idea behind~\FPP before describing our extension. 

\smallskip\noindent{\bfseries Algorithm~\FPP.} Let~$G$ be a maximal plane graph 
with~$n \geq 3$ vertices. The \FPP algorithm first computes a suitable linear 
ordering of the vertices of~$G$, called a \emph{canonical ordering} of~$G$,  
and then incrementally constructs a drawing of~$G$ using a technique called 
\emph{shift method}. This method adds one vertex per time, following the 
computed canonical ordering and shifting vertices already in the drawing when 
needed. Namely, let $\sigma=(v_1,v_2,\dots,v_n)$ be a linear ordering of the 
vertices of~$G$. For each integer~$k \in [3, n]$, denote by~$G_k$ the plane 
subgraph of~$G$ induced by the~$k$ vertices~$v_1,v_2,\dots,v_k$ ($G_n=G$) and 
by~$C_k$ the boundary of the outer face of~$G_k$, called the \emph{contour} 
of~$G_k$. Ordering~$\sigma$ is a canonical ordering of~$G$ if the following 
conditions hold for each integer~$k \in [3, n]$: 
\begin{enumerate*}[label=(\roman{*})]
  \item $G_k$ is biconnected and internally triangulated;
  \item $(v_1,v_2)$ is an outer edge of~$G_k$; and 
  \item if~$k+1 \leq n$, vertex~$v_{k+1}$ is located in the outer face 
    of~$G_k$, and all neighbors of~$v_{k+1}$ in~$G_k$ appear on~$C_k$ 
    consecutively.
\end{enumerate*}

We call \emph{lower neighbors} of~$v_k$ all neighbors~$v_j$ of~$v_k$ for 
which~$j < k$. Following the canonical ordering~$\sigma$, the shift method 
constructs a drawing of~$G$ one vertex 
per time. The drawing~$\Gamma_k$ computed at step~$k$ is a drawing of~$G_k$. 
Throughout the computation, the following invariants are maintained for 
each~$\Gamma_k$, with~$3 \leq k \leq n$:
\begin{enumerate*}[label=({\bfseries I\arabic*})]
  \item \label{inv1} $p_{v_1}=(0,0)$ and $p_{v_2}=(2k-4,0)$;
  \item \label{inv2} $x(w_1)<x(w_2)<\dots<x(w_t)$, where 
    $w_1=v_1, w_2,\dots,w_t=v_2$ are the
    vertices that appear along~$C_k$, going from~$v_1$ 
    to~$v_2$.
  \item \label{inv3} Each edge~$(w_i,w_{i+1})$ (for $i=1,2,\dots,t-1$) is drawn 
    with slope either~$+1$ or~$-1$.
\end{enumerate*}

More precisely,~$\Gamma_3$ is constructed placing~$v_1$ at~$(0,0)$, $v_2$ 
at~$(2,0)$, and~$v_3$ at~$(1,1)$. The 
addition of~$v_{k+1}$ to~$\Gamma_k$ is executed as follows. Let 
$w_p,w_{p+1},\dots,w_ q$ be the lower neighbors of~$v_{k+1}$ ordered from left 
to right. Denote by~$\mu(w_p,w_q)$ the intersection point between the line 
with slope~$+1$ passing through~$w_p$ and the line with slope~$-1$ passing 
through~$w_q$. Point~$\mu(w_p,w_q)$ has integer coordinates and thus it 
is a valid placement for~$v_{k+1}$. With this placement, 
however,~$(v_{k+1},w_p)$ and~$(v_{k+1},w_q)$ may overlap 
with ~$(w_p,w_{p+1})$ and~$(w_{q-1},w_q)$, respectively; see 
Figure~\ref{fi:shift-1}. To avoid this, a
\emph{shift} operation is applied: $w_{p+1}$, $w_{p+2}$,$\dots$,$w_{q-1}$ 
are shifted to the right by~1 unit, and $w_q, w_{q+1}, \dots, w_t$ 
are shifted to the right by~2 units. Then~$v_{k+1}$ is placed at 
point~$\mu(w_p,w_q)$ with no overlap; see Figure~\ref{fi:shift-2}.
We recall that, to keep planarity, when the algorithm 
shifts a vertex~$w_i$ ($p+1 \leq i \leq t$) of~$C_k$, it also shifts some of the
inner vertices together with it; for more details on this point refer 
to~\cite{cp-ltadp-IPL95,dpp-hdpgg-C90}. By Invariants~\ref{inv1} and~\ref{inv3},
the area of the final drawing is $(2n-4) \times (n-2)$.

\begin{figure}[t]
\centering
\subfigure[]{\includegraphics[scale=0.65,page=1]{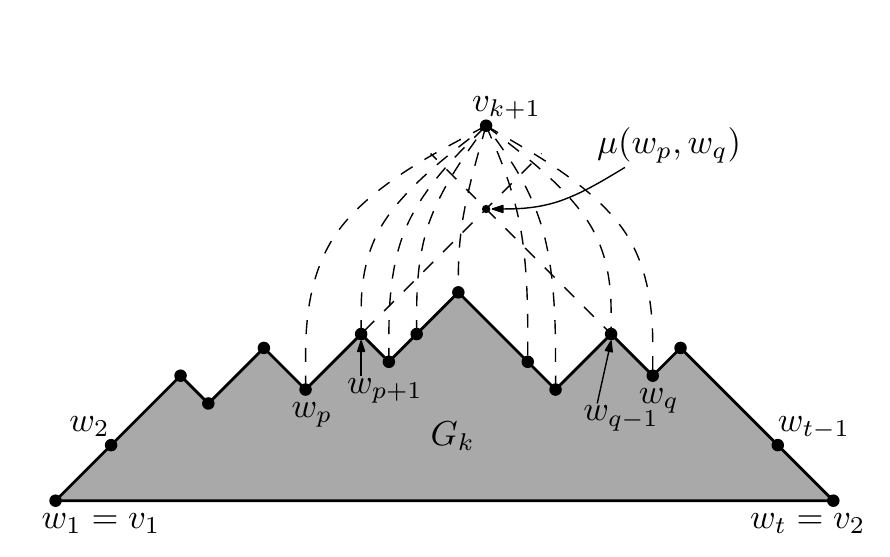}\label{fi:shift-1}}\hfil
\subfigure[]{\includegraphics[scale=0.65,page=2]{shift}\label{fi:shift-2}}
  \caption{Illustration of the shift algorithm at the addition step of 
    $v_{k+1}$. The shift operation changed the slopes of the edges drawn bold. (a) Placing $v_{k+1}$ at $\mu(w_p,w_q)$ would create overlaps. (b) After the shift operation, $v_{k+1}$ can be placed at 
      $\mu(w_p,w_q)$ without overlaps. }
  \label{fi:shift}
\end{figure}

\medskip\noindent{\bfseries Our extension.} Let~$G$ be an IC-plane graph, and 
assume that~$G^+$ is the plane-maximal IC-plane graph obtained from~$G$ by 
applying the technique of Lemma~\ref{lem:3-connected}. Our drawing algorithm 
computes an IC-planar drawing of~$G^+$ with right angle crossings, by 
extending algorithm \FPP. It adds to the classical shift operation \emph{move}
and \emph{lift} operations to guarantee that one of the crossing edges of a kite
is vertical and the other is horizontal. We now give an idea of our technique,
which we call \EFPP. Details are given in the proof of Theorem~\ref{th:rac-drawings}.
Let~$\sigma$ be a canonical ordering constructed from the underlying maximal 
plane graph of~$G^+$. 
Vertices are incrementally added to the drawing, according to~$\sigma$, 
following the same approach as for \FPP. However, suppose that $K =(a,b,c,d)$ 
is a kite of~$G^+$, and that~$a$ and~$d$ are the first and the last vertex 
of~$\sigma$ among the vertices in~$K$, respectively. Once~$d$ has been added to 
the drawing, the algorithm applies a suitable combination of move and lift 
operations to the vertices of the kite to rearrange their positions so to 
guarantee a right angle crossing. Note that, following the \FPP technique,~$a$
was placed at a $y$-coordinate smaller than the $y$-coordinate of~$d$. A 
move operation is then used to shift~$d$ horizontally to the same 
$x$-coordinate as~$a$ (i.e.,~$(a,d)$ becomes a vertical segment in the drawing); 
a lift operation is used to vertically shift the lower between~$b$ and~$c$, 
such that these two vertices get the same $y$-coordinates.
Both operations are applied so to preserve planarity and to maintain 
Invariant~\ref{inv3} of \FPP;  however, they do not maintain Invariant~\ref{inv1}, 
thus the area can increase more than in the \FPP algorithm and 
may be exponential. The application of move/lift operations on the vertices 
of two distinct kites do not interfere each other, as the kites do not share 
vertices in an IC-plane graph. The main operations of the algorithm are depicted
in Figure~\ref{fi:lift}.

\newcommand{\thRacDrawings}{Let $G$ be an IC-plane graph with $n$ vertices. There exists an $O(n^3)$-time algorithm that constructs a straight-line IC-plane \Rac grid drawing of $G$.}

\begin{theorem}\label{th:rac-drawings}
  \thRacDrawings
\end{theorem}
\begin{proof} 
  Let $G^+$ be the augmented graph constructed from~$G$ by using
  Lemma~\ref{lem:3-connected}. Call~$G'$ the subgraph obtained from~$G^+$ by 
  removing one edge from each pair of crossing edges;~$G'$ is a maximal plane 
  graph (see condition~\ref{main-triangulated} of Lemma~\ref{lem:3-connected}). 
  We apply on~$G'$ the shelling procedure used by de Fraysseix {\em et al.} 
  to compute a canonical ordering~$\sigma$ of~$G'$ in~$O(n)$ 
  time~\cite{dfpp-sssfe-STOC88}; it goes backwards, starting from a vertex on 
  the outer face of~$G'$ and successively removing a vertex per time from the 
  current contour. However, during this procedure, some edges of~$G'$ can be 
  replaced with some other edges of~$G^+$ that were previously excluded, 
  although~$G'$ remains maximal planar. Namely, whenever the shelling procedure
  encounters the first vertex~$d$ of a kite~$K=(a,b,c,d)$, it marks~$d$ 
  as~$\ktop(K)$, and considers the edge~$e$ of~$K$ that is missing in~$G'$. 
  If~$e$ is incident to~$d$ in~$K$, the procedure reinserts it and 
  removes from~$G'$ the other edge of~$K$ that crosses~$e$ in~$G^+$. 
  If~$e$ is not incident to~$d$, the procedure continues without varying~$G'$.
  We say that $u\prec v$ if $\sigma(u)<\sigma(v)$.
  
  We then compute a drawing of~$G^+$ by using the \EFPP algorithm. 
  Let vertex~$v=v_{k+1}$ be the next vertex to be placed according to~$\sigma$.
  Let~$\U(v)$ be the set of lower neighbors of~$v$, and let~$\Left(v)$ 
  and~$\Right(v)$ be the leftmost and the rightmost vertex in~$\U(v)$, 
  respectively. Also, denote by~$\A_l(v)$ the vertices to the top-left of~$v$, 
  and by~$\A_r(v)$ the vertices to the top-right of~$v$. If~$v$ is 
  not~$\ktop(K)$ for some kite~$K$, then~$v$ is placed by following the 
  rules of \FPP, that is, at the intersection of the~$\pm 1$ diagonals 
  through~$\Left(v)$ and~$\Right(v)$ after applying a suitable shift operation. 
  If~$v = \ktop(K)$ for some kite~$K$, the algorithm proceeds as follows. 
  Let~$K = (a,b,c,d)$ with~$v = d = \ktop(K)$. The remaining three 
  vertices of~$K$ are in~$G_k$ and are consecutive along the contour~$C_k$, 
  as they all belong to $\U(d)$ (by construction, $G'$ contains edge~$(a,d)$).
  W.l.o.g., assume that they are encountered in the order~$\{b,a,c\}$ from left
  to right.  The following cases are now possible: 
  
  \begin{figure}[t]
\centering
\subfigure[]{\includegraphics[scale=0.5,page=1]{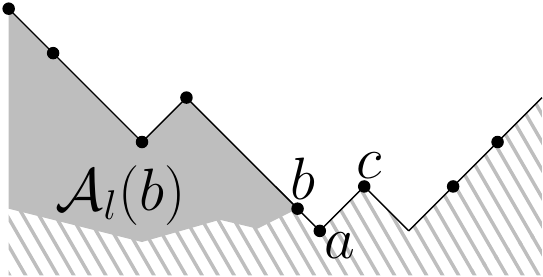} \label{fi:lift-a}}\hfil
\subfigure[]{\includegraphics[scale=0.5,page=2]{rac-undo}\label{fi:lift-b}}
\subfigure[]{\includegraphics[scale=0.5,page=3]{rac-undo} \label{fi:move-a}}\hfil
\subfigure[]{\includegraphics[scale=0.5,page=4]{rac-undo}\label{fi:move-b}}
\caption{(a-b) The lift operation: (a) Vertex $b$ is $r$ units below $c$. (b) Lifting $b$. (c-d) The move operation: (c) Vertex $d$ is $s$ units to the left of $b$. (d) Moving $d$. }
\label{fi:lift}
\end{figure}

  \smallskip\noindent{\bfseries Case 1:} $a \prec b$ and~$a \prec c$. This 
  implies that~$a = \Right(b)$ and~$a = \Left(c)$. The edges~$(a,b)$ and~$(a,c)$ 
  have slope~$-1$ and~$+1$, respectively, as they belong to~$C_k$. We now 
  aim at having~$b$ and~$c$ at the same $y$-coordinate, by applying a lift 
  operation. Suppose first that~$r = y(c) - y(b)>0$; see Figure~\ref{fi:lift-a}. 
  We apply the following steps:
  \begin{enumerate*}[label=(\roman{*})]
    \item Temporarily undo the placement of~$b$ and of all vertices 
      in~$\A_l(b)$. 
    \item Apply the shift operation to vertex~$\Right(b)=a$ by~$2r$ units to the
      right, which implies that the intersection of the diagonals 
      through~$\Left(b)$ and~$\Right(b)$ is moved by~$r$ units to the right and 
      by~$r$ units above their former intersection point. Hence,~$b$ and~$c$ are
      placed at the same $y$-coordinate; see also Figure~\ref{fi:lift-b}. 
    \item Reinsert the vertices of~$\A_l(b)$ and modify~$\sigma$ accordingly. 
      Namely, by definition, each vertex in~$\A_l(b)$ does not belong to~$\U(b)$
      and it is not an inner vertex below~$b$; therefore, vertices in~$\A_l(b)$
      can be safely removed. Hence,~$\sigma$ can be modified such that~$b\prec w$
      for each $w\in A_l(b)$.
  \end{enumerate*}
    If~$r = y(c) - y(b)<0$, a symmetric operation is applied:
  \begin{enumerate*}[label=(\roman{*})] 
    \item Undo the placement of~$c$ and of all vertices in~$\A_r(c)$. 
    \item Apply the shift operation to vertex~$\Right(c)$ by~$|2r|$ additional 
      units to the right. 
    \item Reinsert the vertices of~$\A_r(c)$. 
  \end{enumerate*}
  
  Finally, we place~$d$ vertically above~$a$. To this aim, we first apply the
  shift operation according to the insertion step of \FPP. After that, we may 
  need to apply a move operation; see Figure~\ref{fi:move-a}. 
  If~$s=x(d)-x(a)>0$, then we apply the shift operation to vertex~$\Right(d)=c$ 
  by~$2s$ units to the right and then place~$d$ (see Figure~\ref{fi:move-b}). 
  If $s=x(d)-x(a)<0$, then we apply the shift operation to vertex~$\Left(d)=b$ by~$2s$ 
  units to the left and then place~$d$ (clearly, the shift operation can be used
  to operate in the left direction with a procedure that is symmetric to the one
  that operates in the right direction). 
  Edges~$(a,d)$ and~$(b,c)$ are now vertical and horizontal, respectively. In 
  the next steps, their slopes do not change, as their endpoints are shifted 
  only horizontally (they do not belong to other kites); also,~$a$ is shifted 
  along with~$d$, as it belongs to~$\U(d)$.

  {\smallskip\noindent\bfseries Case 2}: $b \prec a \prec c$ or $c \prec a \prec b$. 
  We describe how to handle the case that $b \prec a \prec c$, as the other case 
  can be handled by symmetric operations. This 
  implies that~$b = \Left(a)$ and~$a = \Left(c)$. The edges~$(b,a)$ and~$(a,c)$ 
  both have slope~$+1$ respectively, as they belong to~$C_k$.  
  Let $\{z_1 \prec \ldots \prec z_r\}$ be the sequence 
  of $r \geq1 $ neighbors of $b$ inside $\A_r(b)$, with~$a = z_r$
  and~$b = \Left(z_i)$, with $1 \leq i \leq r$, as shown in Figure~\ref{fi:triang_sep}. 
  
  \begin{figure}[t]
    \centering
    \includegraphics[scale=0.8]{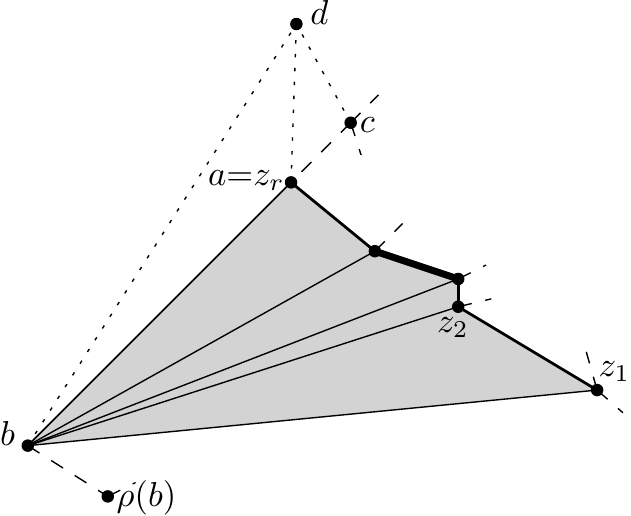}
    \caption{Illustration for the proof of Theorem.~\ref{th:rac-drawings}. The edge 
      with slope $\alpha_{min}$ is thicker.}\label{fi:triang_sep}
  \end{figure}
  
  Consider the slopes~$\alpha_i$ of the 
  edges~$(z_i, z_{i+1})$ for~$1 \leq i < r$, 
  and let~$\alpha_{\min}$ be the negative slope with the
  least absolute value among them; see the bold edge in Figure~\ref{fi:triang_sep}. 
  Let $s$ be the (negative) slope of the edge~$(b, \Right(b))$. We aim at 
  obtaining a drawing where~$|s|\leq |\alpha_{\min}|$. To this aim, we apply the 
  shift operation on~$\Right(b)$ by~$x$ units to the right, which stretches and 
  flattens the edge~$(b, \Right(b))$. If~$|\alpha_{\min}| = h/w$, 
  and~$|s| =  h'/w'$, then the value of~$x$ is the first even integer such 
  that $x \geq (h'w - hw')/h$. Now we have that 
  $|s|= h'/(w'+x)\le h'/(w'-h'w/h-w')=h/w=|\alpha_{\min}|$.
  The fact that~$x$ is even preserves the even 
  length of the edges on the contour. This preliminary operation will be useful 
  in the following.

  Next, let~$\Delta(b) = y(b) - y(\Right(b)) >0$
  and let~$\Delta(c) = y(c) - y(b)>0$ , i.e.,~$b$ lies~$\Delta(b)$ rows 
  above~$\Right(b)$, and~$c$ lies~$\Delta(c)$ rows above~$b$, where the 
  edges~$(b, a)$ and~$(a, c)$ have slope~$+1$. We apply the following procedure 
  to lift $b$ at the same $y$-coordinate of $c$. 
  \begin{enumerate*}[label=(\roman*)]
    \item We undo the placement of all vertices in $\A_l(b)$. 
    \item If~$\Delta(c)$ is not a multiple of~$\Delta(b)$, 
      say $\Delta(b) + \delta = q \cdot \Delta(c)$ for some integer~$q$, then we
      shift $\Right(c)$ by~$2\delta$ units to the right. This
      implies that~$c$ moves by~$(\delta,\delta)$ above its former position. 
    \item We set the $y$-coordinate of vertex~$b$ equal to the $y$-coordinate 
      of~$c$. To that end, we stretch the edge~$(\Right(b),b)$ by the 
      factor~$q$. Let~$w'$ be the width and~$h'$ be the height of the
      edge~$(\Right(b),b)$. The new edge has the same slope as before, and has 
      width~$q w'$ and height~$q h'$. This implies shifting all
      vertices $\Right(b),z_1,\ldots,z_{r-1},a,c$ by~$(q-1) w'$ units to the 
      right. Vertex~$b$ may need a further adjustment by a single unit left 
      shift if~$b= \Left(d)$ and the intersection point of the $\pm 1$ diagonals
      through~$\Left(d)$ and~$\Right(d)$ is not a grid point.  Also, we apply 
      the shift operation on~$\Left(b)$ by~$(q-1) h'$ units to the left. This 
      particular lifting operation applied on vertex $b$ preserves planarity, 
      which could be violated only by edges incident to~$b$.
      Namely, if vertex~$w$ is a neighbor of~$b$ in~$\U(b)$, then the
      edge~$(w,b)$ is vertically stretched by~$(q-1)h'$ units.
      This cannot enforce a crossing, since it means a vertical shift of~$w$. 
      Clearly,~$b$ can see~$\Right(b)$, since the edge was stretched.
      Consider the upper right neighbors~$z_1, \ldots, z_r$ with~$z_r=a$ of~$b$.
      The edges~$(b, z_i)$ change direction from right upward to right
      downward. Since the absolute value of the slope of the 
      edge~$(b,\Right(b))$ is bounded from above by $\alpha_{\min}$ and 
      since~$y(\rho(b))\le y(b)\le y(z_i); 1\le i\le r$, the new
      position of~$b$ is below the line spanned by each edge~$(z_i, z_{i+1})$ 
      for~$1 \leq i <r$. Hence,~$b$ can see each such neighbor~$z_i$, 
      including~$a = z_r$.  The lifting of~$b$ has affected all 
      vertices~$v \in \A_l(b)$ with~$y(v)<y(d)$. 
    \item We re-insert the vertices in $\A_l(b)$, by changing $\sigma$
      accordingly, as already explained for {\bfseries Case 1}. 
  \end{enumerate*}  

  Finally, we place~$d = \ktop(K)$. First, we place~$d$ at the intersection
  point of the $\pm 1$ diagonals through~$\Left(d)$ and~$\Right(d)$. Then, we 
  adjust~$d$ such that it lies vertically above~$a$. If the preliminary position
  of~$d$ is~$t$ units to the left (right) of~$a$, then we apply  the shift 
  operation on~$\Right(d)$ by $2t$ units to the right (on $\Left(d)$ by $2t$ 
  units to the left). 
    
  {\smallskip\noindent\bfseries Case 3}: $b \prec a$ and $c \prec a$. 
  This implies that~$b = \Left(a)$ and~$c = \Right(c)$. The edges~$(b,a)$ and~$(c,a)$ 
  have slope~$+1$ and~$-1$, respectively, as they belong to~$C_k$. We now 
  aim at having both~$b$ and~$c$ at the $y$-coordinate $y(a)+1$.
  To this end, we use the procedure described in \textbf{Case 2}
  to lift~$b$ upwards by $y(a)+1-y(b)$ rows by using a dummy vertex $c'$ at
  position $(x(c),y(a)+1)$ as a reference point. Note that this lifting only 
  affects~$b$ and the vertices in~$\A_l(b)$ (all other vertices are moved 
  uniformly), so both~$a$ and~$c$ remain at their position. Hence, we now have
  the situation $c \prec a \prec b$ and can again use the procedure of 
  \textbf{Case 2} to solve this case.

  \smallskip To conclude the proof, we need to consider the first edge of the construction
  which is drawn horizontal. Since the lift operation requires an edge that does 
  not have slope~0, we may need to introduce dummy vertices and edges. 
  Namely, if there is a kite including the base edge~$(v_1,v_2)$, then we add 
  two dummy vertices~$1',2'$ below it that form a new base edge~$(v_1',v_2')$. 
  We add the additional edges $(v_1',v_1)$, $(v_1',v_2)$, $(v_1',v_n)$, 
  $(v_2',v_2)$ and $(v_2',v_n)$ to make the graph maximal planar. These dummy 
  vertices and edges will be removed once the last vertex $v_n$ is placed.
  
  In terms of time complexity, $G^+$ can be computed in $O(n)$ time, by
  Lemma~\ref{lem:3-connected}. Furthermore, each shift, move and the lift 
  operation can be implemented in $O(n)$ time, hence the placement of a single
  vertex costs $O(n)$ time. However, in some cases (in particular, when we are 
  placing the top vertex of a kite), we may need to undo the placement of a set 
  of vertices and re-insert them afterwards. Since when we undo and reinsert a 
  set of vertices we also update $\sigma$ accordingly, this guarantees that the 
  placement of the same set of vertices will not be undone anymore. Thus, the 
  reinsertion of a set of $O(n)$ vertices costs $O(n^2)$. Hence, we 
  have $\sum^n_{i=1} O(n + n^2)$ which gives an $O(n^3)$ time complexity.
\end{proof}
  
Figure~\ref{fi:dfpp} shows a running example of our algorithm.
Theorem~\ref{th:rac-drawings} and the fact that there exist $n$-vertex \Rac graphs 
with~$4n-10$ edges~\cite{del-dgrac-2011} while an $n$-vertex IC-planar graph 
has at most~$13n/4-6$ edges~\cite{zl-spgic-CEJM13} imply
that IC-planar graphs are a proper subfamily of \Rac graphs.


\begin{figure}[h!]
  \centering
  \subfigure[]{\includegraphics[scale=0.55,page=1]{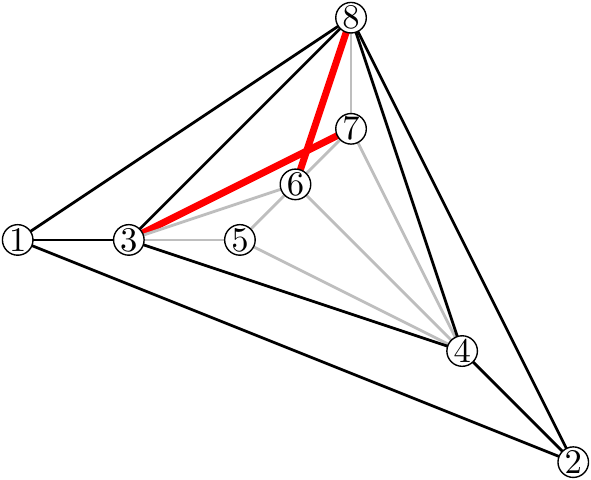}\label{fi:dfpp-1}}\hfill
  \subfigure[]{\includegraphics[scale=0.55,page=2]{dfpp}\label{fi:dfpp-2}}\hfill
  \subfigure[]{\includegraphics[scale=0.55,page=3]{dfpp}\label{fi:dfpp-3}}
  \subfigure[]{\includegraphics[scale=0.6,page=4]{dfpp}\label{fi:dfpp-4}}
  \subfigure[]{\includegraphics[scale=0.75,page=5]{dfpp}\label{fi:dfpp-5}}
  \caption{Example run of our algorithm on an IC-planar graph~$G$ with a 
    separating triangle. The crossing edges are drawn bold, the edges inside
    the separating triangle are drawn gray. 
    (a) Input graph~$G$. 
    (b) Output of \FPP after vertex~7. 
    (c) Output of \FPP after vertex~8. 
    (d) Lifting~3 to the level of~7. 
    (e) Moving~8 directly above~6.}
  \label{fi:dfpp}
\end{figure}

We now show that exponential area is required for \Rac drawings of IC-planar 
graphs. Since the vertices are not drawn on the integer grid, the drawing area 
is measured as the proportion between the longest and the shortest edge.

\newcommand{\thRacArea}{
For every integer $k \geq 1$, there exists an IC-plane graph $G_k$ with $n_k$ vertices such that every IC-planar straight-line \Rac drawing of $G_k$ takes area $\Omega(q^{n_k})$, for some constant $q>1$.  
  }

\begin{theorem}\label{th:rac-area}
  \thRacArea
\end{theorem}
\begin{proof}
Refer to Figures~\ref{fi:exponential-1} and~\ref{fi:exponential-2} for the construction of $G_k$, for $k \geq 1$. Each graph $G_i$, for $1 \leq i \leq k$, has a 4-cycle as outer face, while all other faces are triangles (a triangle is composed of either three vertices or of two vertices and one crossing point). Two non-adjacent edges of the outer face are called the \emph{marked} edges of $G_i$.

\begin{figure}[t]
\centering
\subfigure[The first three levels of the construction.]{\includegraphics[page=1]{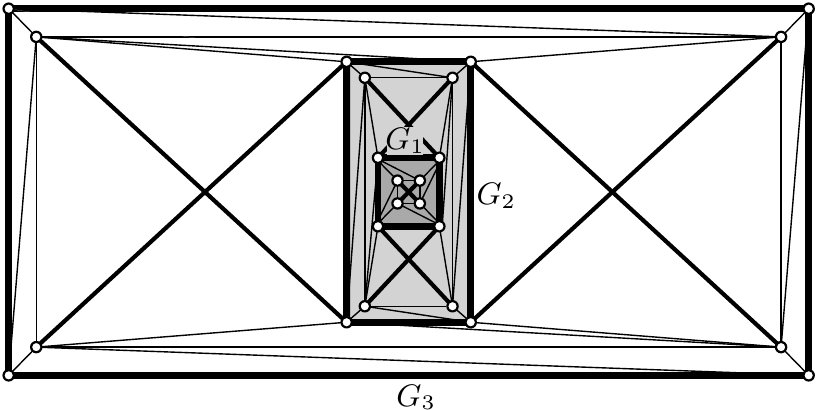}\label{fi:exponential-1}}
\subfigure[Going from level $i-1$ to level $i$.]{\includegraphics[page=2]{area-G-2}\label{fi:exponential-2}}
\caption{Illustration for the proof of Theorem~\ref{th:rac-area}. }\label{fi:exponential}
\end{figure} 

Graph~$G_1$ has~8 vertices,~4 inner vertices forming a kite, and~4 vertices 
on the outer face. All inner faces are triangles. The two marked edges of the 
outer face of~$G_1$ are any two non-adjacent edges of this face.  Graph~$G_{i}$ 
is constructed from~$G_{i-1}$ as follows, see also Figure~\ref{fi:exponential-2}. 
Let $\{A,B,C,D\}$ be the four vertices of the outer face of~$G_{i-1}$, and 
let~$(A,D)$ and~$(B,C)$ be the two marked edges of~$G_{i-1}$ (bold in 
Figure~\ref{fi:exponential-2}). We attach a kite $(A,D,a,d)$ on the marked 
edge~$(A,D)$ and a kite $(B,C,b,c)$ on the marked edge~$(B,C)$. We connect the 
two kites with the edges~$(a,B)$, $(a,b)$, $(c,D)$, and $(c,d)$. We 
then add a cycle between four further vertices $\{\alpha,\beta,\gamma,\delta\}$ 
that form the outer face of $G_i$, and we triangulate the inner face between 
the cycles $(\alpha,\beta,\gamma,\delta)$ and $(a,b,c,d)$. We set the 
edges $(\alpha,\beta)$ and $(\gamma,\delta)$ as the marked edges of $G_i$.

An embedding-preserving straight-line \Rac drawing $\Gamma_k$ 
of $G_k$ ($k \geq 1$) with minimum area can be obtained by drawing each kite as 
a quadrilateral, as shown in Figures~\ref{fi:exponential-1} and~\ref{fi:exponential-2}. 
Since vertex~$a$ is connected to vertex~$B$, $a$ has to lie above the line 
spanned by the edge $(A,B)$. Further, since vertex~$d$ is connected to vertex~$C$, 
$d$ has to lie below the line spanned by the edge $(C,D)$. This implies that
the quadrilateral $(A,D,a,d)$ contains the square that has the edge~$(A,D)$
as its right side. Hence, the width increases by at least the length of edge~$(A,D)$.
Note that one might extend the edge~$(B,C)$ to make~$(A,D,a,d)$ 
smaller. However, since we also have the edges~$(a,b)$ and~$(c,d)$, this 
procedure increases the size of~$(B,C,b,c)$ such that, as soon as $(a,d)$ is
smaller than~$(A,D)$, it has to be that~$(b,c)$ is larger than~$(B,C)$
(to keep the visibility required for~$(a,b)$ and~$(c,d)$).
Then, this implies that the quadrilateral $(B,C,b,c)$ contains the square that 
has the edge~$(B,C)$ as its left side, so the width increases by at least the 
length of edge~$(B,C)$.

The minimum-area drawing forces the outer face of every subgraph 
$G_i$ ($i \leq k$) of~$G_k$ to be a rectangle $R_i$. We denote by $l_i$  ($L_i$) 
the length of the shortest (longest) side of $R_i$. The area of $R_i$ is 
$A_i = l_i \times L_i$. Observe that, by construction, the marked edges of $G_i$ 
have length $L_i$. It follows that $l_{i} \ge L_{i-1} +4$ and $L_{i}\ge l_{i-1} + L_{i-1}+2$. Therefore, 
$A_{i} \ge (L_{i-1} +4) \times (l_{i-1} + L_{i-1}+2) \geq A_{i-1} +A_{i-1} 
= 2 A_{i-1}$, hence $A_k \geq 2 A_{k-1} \geq 2^{k-1} A_1 \geq 2^{k+1}$. 
The number of vertices of $G_k$ is $n_k = 8k$, and hence $k = \frac{n_k}{8}$. Thus $A_k  \geq 2^{\frac{n_k}{8}}\geq 1.09^{n
_k}$, which proves the statement.
\end{proof}

\section{Recognizing IC-planar graphs}\label{se:recognition}

The \emph{IC-planarity testing} problem asks if a graph~$G$ admits an IC-planar embedding.
\medskip

{\noindent \bfseries Hardness of the problem.}  The next theorem shows that IC-planarity testing is \NP-complete. 

\newcommand{\thNpVar}{IC-planarity testing is \NP-complete.}
\begin{theorem}\label{th:np-hard}
  \thNpVar
\end{theorem}

\begin{figure}[tb]
\centering
\subfigure[$G$]{\includegraphics[scale=1]{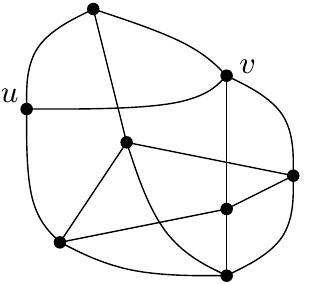}\label{fig:1planar-instance}}\hfil
\subfigure[$G^*$]{\includegraphics[scale=1]{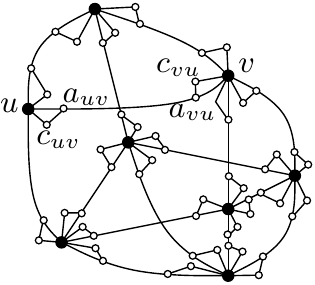}\label{fig:ICplanar-instance}}
\caption{Illustration of the proof of Theorem~\ref{th:np-hard}.}
\label{fi:reduction-gadget}
\end{figure}

\begin{proof}
  IC-planarity is in \NP, as one can guess an embedding and
  check whether it is IC-planar~\cite{gj-1983}. For the hardness proof, the reduction is from  
  the \emph{1-planarity testing} problem, which asks whether a given graph is 1-planar or not.
  The reduction uses a $3$-cycle gadget and exploits the fact that at most one edge of 
  a $3$-cycle is crossed in an IC-planar drawing.  We transform an instance~$G$ of 1-planarity testing 
	into an instance~$G^*$ of IC-planarity testing, by
  replacing each edge~$(u,v)$ of $G$ with a graph~$G_{uv}$ 
	consisting of two $3$-cycles,~$T_{uv}$ and~$T_{vu}$, with vertices~$\{u, c_{uv}, a_{uv}\}$ 
  and~$\{v, c_{vu}, a_{vu}\}$, respectively, plus edge~$(a_{uv},a_{vu})$, called the \emph{attaching edge} of~$u$ 
  and~$v$; see Figure~\ref{fi:reduction-gadget}. 
  
  Let~$\Gamma$ be a 1-planar drawing of~$G$. An IC-planar drawing~$\Gamma^*$ 
	of~$G^*$ can be easily constructed by replacing each curve representing an 
	edge~$(u,v)$ in~$\Gamma$ with a drawing of~$G_{uv}$ where~$T_{uv}$ 
  and~$T_{vu}$ are drawn planar and sufficiently small, such that the possible
  crossing that occurs on the edge~$(u,v)$ in~$\Gamma$ occurs on the attaching
  edge~$(a_{uv},a_{vu})$ in~$\Gamma^*$. Hence, since all the attaching edges are
  independent,~$\Gamma^*$ is IC-planar.

	Let~$\Gamma^*$ be an IC-planar drawing of~$G^*$. We show that it is possible 
	to transform the drawing in such a way that all crossings occur only between 
	attaching edges. Once this condition is satisfied, in order to construct a
	1-planar drawing~$\Gamma$ of~$G$, it suffices to remove, for each 
  edge~$(u,v)$, the vertices~$c_{uv}$ and~$c_{vu}$, and to replace~$a_{uv}$
  and~$a_{vu}$ with a bend point. Namely, as already observed, no more than one 
  edge can be crossed for every gadget~$T_{uv}$ of~$G^*$.  Suppose now that the
  edge~$(u,a_{uv})$ of~$T_{uv}$ is crossed. Since the other two edges 
  of~$T_{uv}$ are not crossed, we can reroute~$(u,a_{uv})$ such that it follows
  the curves that represent~$(u,c_{uv})$ and~$(c_{uv},a_{uv}$); see 
  Figure~\ref{fi:reduction-reroute-a}.
  
  \begin{figure}[tb]
    \centering
    \subfigure[]{\includegraphics[scale=0.8]{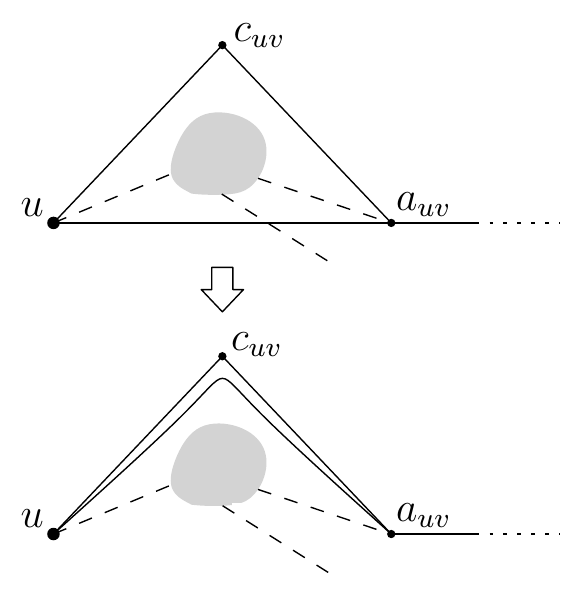}\label{fi:reduction-reroute-a}}
    \subfigure[]{\includegraphics[scale=0.8]{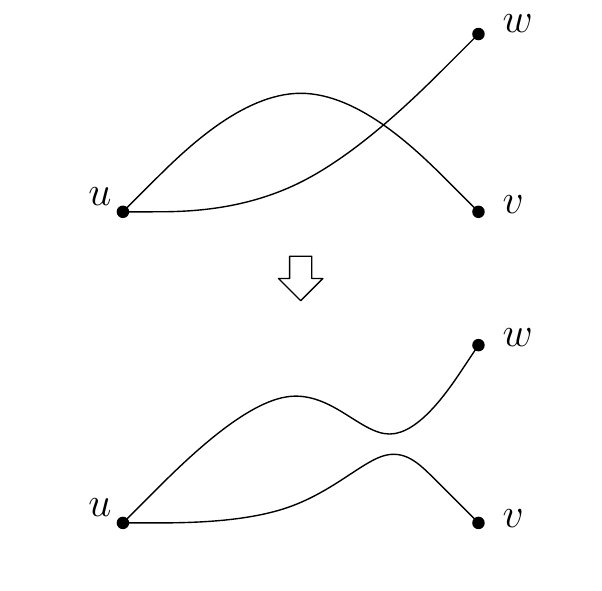}\label{fi:reduction-reroute-b}}
    \caption{Illustration of the proof of Theorem~\ref{th:np-hard}. 
      (a) Rerouting the crossed edge $(u,a_{uv})$ via~$c_{uv}$ to be planar. 
      (b) Rerouting the crossing edges~$(u,v)$ and~$(u,w)$ to be planar.}
    \label{fi:reduction-reroute}
  \end{figure}
  
  In order to complete the proof, we need to take care of the following 
	particular case. Two attaching edges~$a_{uv}$ and~$a_{uw}$ that cross and that 
	are connected to two gadgets~$T_{uv}$ and~$T_{uw}$ with a common vertex~$u$
	represent a valid configuration in~$\Gamma^*$, while they give rise to a 
	crossing between two adjacent edges in~$\Gamma$, which is not allowed 
	since~$\Gamma$ must be a simple drawing. However, this case can be easily 
	solved, since the two edges do not cross any others, by rerouting them in~$\Gamma$ as shown in 
	Figure~\ref{fi:reduction-reroute-b}.
\end{proof}

\smallskip Note that this construction does not work for IC-planarity 
testing with a given rotation system since the rerouting step changes
the rotation system. However, we now prove \NP-hardness of IC-planarity testing for graphs
with a given rotation system. We rely on the membrane technique introduced by Auer {\em et
al.}~\cite{JGAA-347} to prove the \NP-hardness of 1-planarity testing for graphs
with a given rotation system. In particular, we design gadgets that make it possible to use the membrane technique in the case of IC-planar graphs. 

First, we replace the U-graphs~\cite{JGAA-347} by  M-graphs, called \emph{mesh graphs}. These graphs have a unique embedding with a fixed rotation system. Namely, an M-graph is a mesh, where cells are filled with two crossing edges, following a checkerboard pattern to ensure independent crossings; see Figure~\ref{fi:m-graph}. To see that with a given rotation system, an M-graph has a unique embedding, observe that each subgraph isomorphic to $K_4$ can be embedded planarly or as a kite, and this is determined by the rotation system~\cite{Kyncl20091676}. The rotation system defined by the drawing in Figure~\ref{fi:m-graph} implies that each subgraph isomorphic to $K_4$ is embedded as a kite, and therefore the embedding of an M-graph is unique. 

Let~$M$ be an M-graph with a given fixed  embedding. At its bottom line,~$M$ has sufficiently many \emph{free} vertices that are not incident with a crossing edge in~$M$. These vertices are consecutively ordered, say from left to right. The edges on the bottom line are not crossed in any IC-planar embedding of~$M$, so they are crossing-free in the given embedding. Finally,~$M$ cannot be crossed by any path from a free vertex.  In what follows, we attach further gadgets to~$M$ by connecting these gadgets to consecutive free vertices. If there are several gadgets, then they are separated and placed next to each other.

\begin{figure}[tb]
\centering
\subfigure[]{\includegraphics[scale=0.3]{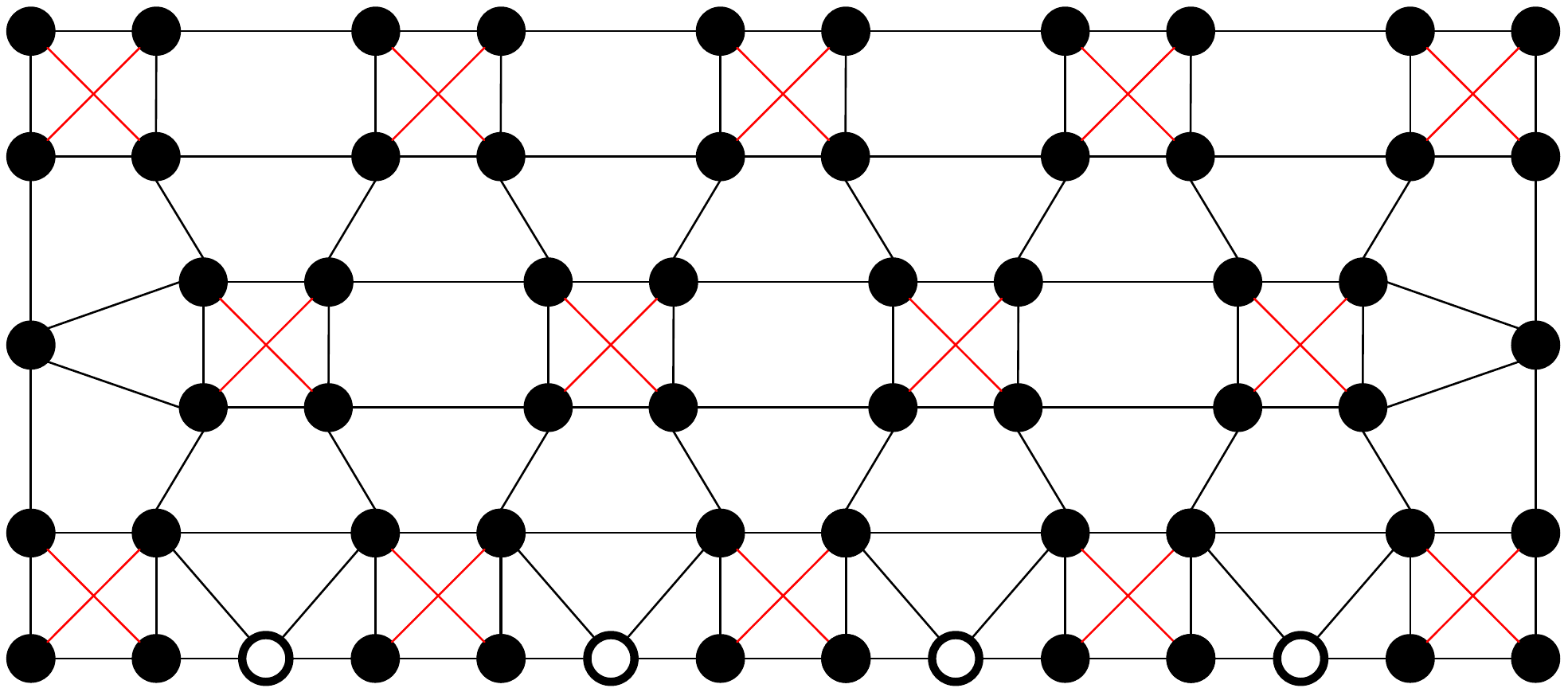}}\hfil
\subfigure[]{\includegraphics[scale=0.4]{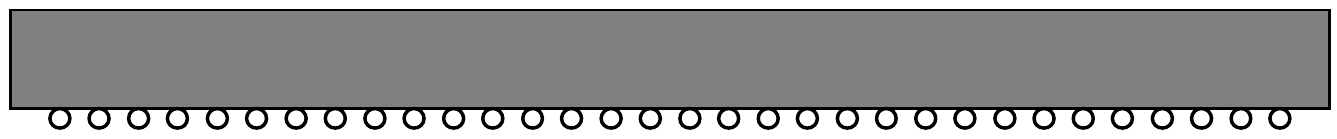}}
\caption{The structure of an M-graph (a) and its abbreviation (b). }
 \label{fi:m-graph}
\end{figure}

\paragraph{General Construction}

Consider an instance~$\alpha$ of planar-3SAT with its corresponding plane
graph~$G$ and its dual~$G^*$.  Recall that the vertices of~$G$ represent
variables~$x$ and clauses~$c$, also, there is an edge~$(x,c)$ if~$x$ or its negation occurs as a
literal in~$c$; see Figure~\ref{fi:planar3sat}. We transform $G^*$ into an 
\emph{M-supergraph} $G^*_{\alpha}$ (see 
Figure~\ref{fi:m-supergraph}) as follows.

Each vertex of~$G^*$, corresponding to a face of the embedded graph~$G$, is 
replaced by a sufficiently large M-graph. Further, each edge of~$G^*$ is 
replaced by a \emph{barrier} of~$l$ parallel edges between~$l$ 
consecutive free vertices on the boundaries of the two M-graphs of the
adjacent faces. These edges will be 
crossed by paths that are called \emph{ropes}. The size of~$l$ will be 
determined later.

For each variable~$x$, we construct a \emph{V-gadget} $\gamma(x)$, and 
similarly we build a \emph{C-gadget} for each clause. These gadgets are described 
below.

\begin{figure}[t]
\centering
\subfigure[]{\includegraphics{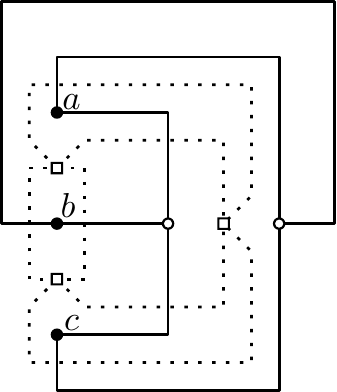}\label{fi:planar3sat}}\hfil
\subfigure[]{\includegraphics[scale=0.4]{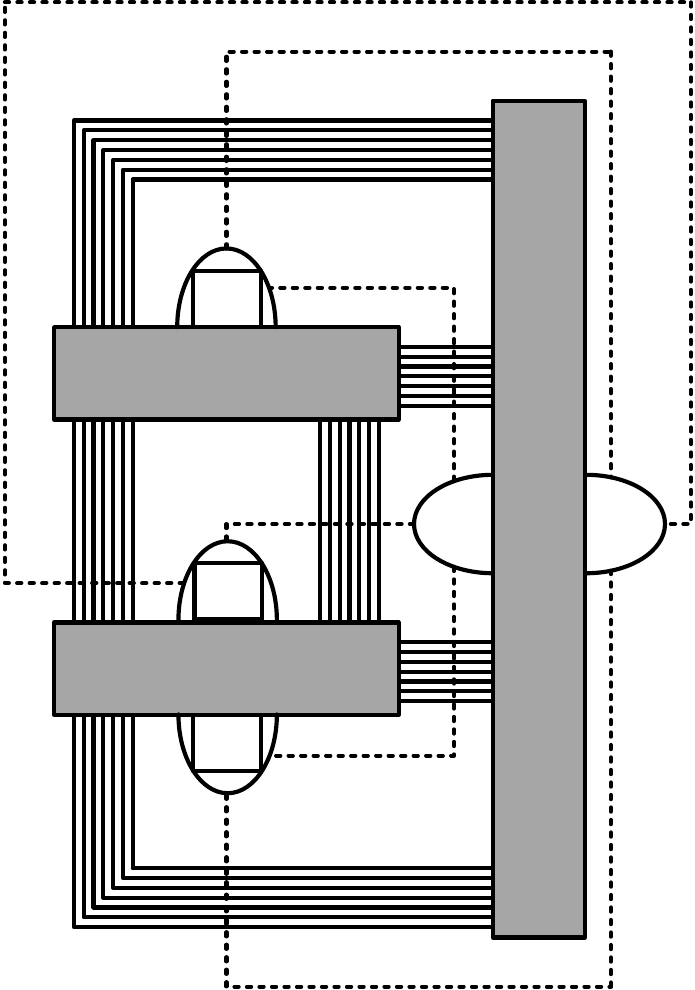}\label{fi:m-supergraph}}
\caption{(a) The planar graph~$G$ (solid) and its dual~$G^*$ (dotted) corresponding to the planar-3SAT formula
    $(a\vee b\vee\neg c)\wedge(a\vee \neg b\vee c)$. (b) The corresponding M-supergraph~$G^*_\alpha$ with the clause gadgets
      (vertical) and the variable gadgets (horizontal).}  
\end{figure}

Each vertex~$u$ of $G$ lies on the boundary of a face~$f$ of~$G$.
We attach the gadget~$\gamma(u)$ of~$u$ to the M-graph at $f$ 
so that~$\gamma(u)$ lies in~$f$. 
It does not matter which face $f$ incident to $u$ is chosen. Similarly,
each edge~$e$ of~$G$ between a variable and a clause is replaced by
a rope of length~$2l + 3$. Since~$e$ is crossed by its dual edge, the rope is 
crossed by a barrier. A rope acts as a communication line that
``passes'' a crossing at a V-gadget across a barrier to a
C-gadget. We denote by a \emph{membrane} (similarly as
in~\cite{JGAA-347}) a path between free vertices on the boundary of a single
M-graph, or between particular vertices of a variable gadget. We call a
vertex \IN if it is placed inside the region of a membrane and the
boundary of the M-graph in an IC-drawing, and \OUT otherwise. 
\IN and \OUT are exactly defined by edges which cross the membrane. 
Note that the framework is basically a simultaneous embedding of~$G$ and~$G^*$
by means of our gadgets, M-graphs, barriers and ropes. The subgraph without 
V- and C-gadgets is 3-connected, since the M-graphs are 3-connected and 
the barriers have size~$l$ for~$l \geq 3$, and it has a unique planar embedding 
if one edge from each pair of crossing edges in each M-graph is removed. 

\paragraph{Construction of the C-gadgets} 

The C-gadget~$c = (l_1, l_2, l_3)$ with three literals~$l_1$,~$l_2$ and~$l_3$
is attached to eight consecutive free boundary vertices of an
M-graph~$M$, say $v_1, \ldots, v_8$. For each literal~$l_i$, there are three 
vertices~$u_i$,~$a_i$ and~$b_i$, and four edges~$(u_i, a_i)$, $(u_i, b_i)$, 
$(a_i, v_{2i})$ and $(b_i,v_{2i+1})$, where~$u_i$ is the initial vertex of the 
rope towards the corresponding variable gadget. There is a membrane of nine 
edges from~$v_1$ to~$v_8$, see Figure~\ref{fi:clause-gadget}.

By construction, at most two vertices among~$u_1$,~$u_2$ and~$u_3$ can be moved 
outside the membrane, and at least one initial vertex of a rope (and maybe
all) must be \IN. \IN shall correspond to the value 
\true~of the literal and thus of the clause.

\paragraph{Construction of the V-gadgets} 

Let~$x$ be a variable and let~$v$ be the vertex corresponding
to~$x$ in~$G$. Suppose that the literal~$x$ occurs in~$k$ clauses
for some~$k \geq 1$, which are ordered by the embedding of~$G$.
Denote this sequence by~$x_1, \ldots x_k$, where each~$x_i$
corresponds to~$x$ or~$\neg x$. The V-gadget of~$x$ is
$\gamma(x) = \gamma(x_0), \gamma(x_1), \ldots, \gamma(x_k),
\gamma(x_{k+1})$. This gadget is connected to~$7(k+2)$ free consecutive vertices 
on the border of the M-graph~$M$ to which it is attached; see 
Figure~\ref{fi:variable-gadget} for an illustration. The gadgets~$\gamma(x_0)$
and~$\gamma(x_{k+1})$  are called the
(left and right) \emph{terminal gadgets} and each~$\gamma(x_i)$ ($1
\leq i \leq k$) is called a
\emph{literal gadget}. Gadget $\gamma_i$ for $0 \leq i \leq k+1$ is similar
to a clause gadget and is connected to seven consecutive free variables~$v^i_1, \ldots, v^i_7$ on the 
boundary of~$M$.
There is a \emph{local membrane} of seven edges from $v^i_1$ to $v^i_7$.

\begin{figure}[tb]
\centering
\subfigure[]{\includegraphics[scale=0.33]{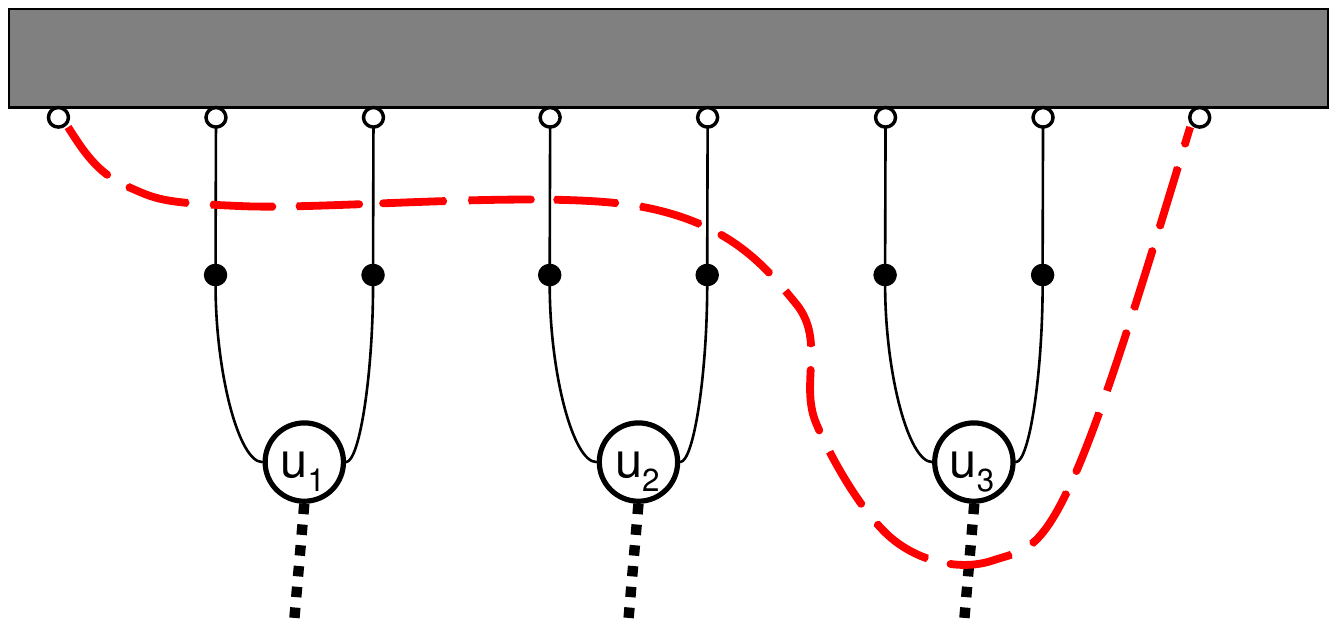}\label{fi:clause-gadget}}\hfil
\subfigure[]{\includegraphics[scale=0.31]{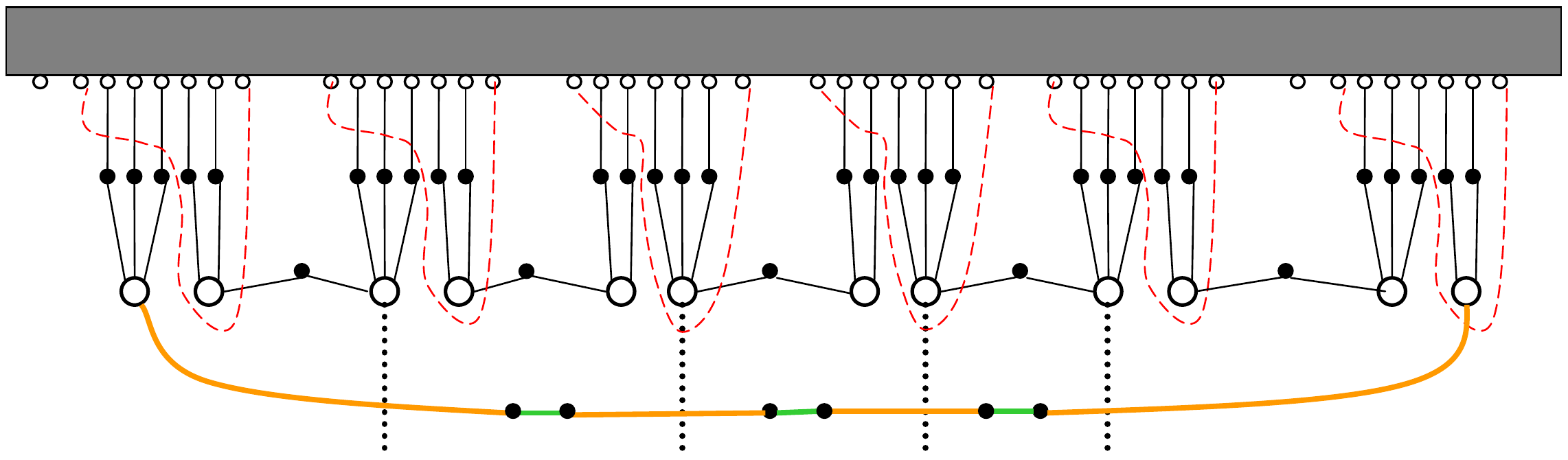}\label{fi:variable-gadget}}
\caption{(a) The clause gadget. (b) The variable gadget.}
\label{fi:gadgets}
\end{figure}

The left terminal gadget~$\gamma(x_0)$ has two primary vertices~$x_0^+$ and~$x_0^-$.
The primary vertex~$x_0^+$ is connected to~$v^0_2$,~$v^0_3$ and~$v^0_4$ by
paths of length two.
The other
primary vertex~$x_0^-$ is connected to $v^0_5, v^0_6$ also by paths of
length two.
Analogously, the right terminal gadget has two primary 
vertices~$x_{k+1}^+$ and~$x_{k+1}^-$, with the same requirements. 
The gadget~$\gamma(x)$ has an \emph{outer membrane} of length~$2k+1$ from~$x_0^+$ 
to~$x_{k+1}^-$.

Each literal gadget~$\gamma(x_i)$ has has two primary vertices~$x_i^+$
and~$x_i^-$. If~$x_i$ is positive, then $x_i^+$ is attached to three
free vertices~$v^i_2$,~$v^i_3$ and~$v^i_4$ of~$M$, and~$x_i^-$
is attached to two free vertices~$v^i_5$ and~$v^i_6$ by two paths of length two,
respectively. The rope to the literal begins at $x_i^+$. Otherwise, if~$x_i$ is
negated, then the gadget is reflected, such that~$x_i^+$ has two, and~$x_i^-$ 
has three paths of length two to the M-graph. The rope to the literal begins
at~$x_i^-$.
The rope is a path of~$2l+3$ edges from vertex~$x_i^{\pm}$ 
of the V-gadget to the vertex~$u_j$ of the clause gadget representing
the literal of $x$.
In addition, there is a path of length two that connects $x_i^-$ to
$x_{i+1}^-$ for all $0 \leq i \leq k$.

The M-graph must have sufficiently many free vertices for 
the edges from the gadgets and barriers. The smallest bound can easily be
computed from the embedding of~$G$ and the attached gadgets; see 
Figure~\ref{fi:m-supergraph}.
The rotation system of the gadgets is retrieved from the drawing and
the ordering of the vertices on the border of M-graphs. 

\paragraph{Correctness}

We will now prove several lemmas on the structure of our construction. With
these lemmas, we will show that an IC-planar drawing of the resulting
graph~$G_\alpha$ immediately yields a valid solution to the underlying 
planar-3SAT problem. First, we show that the M-graph is not crossed.

\begin{lemma}
	The boundary edges of an M-graph (with a fixed rotation system) are
	never crossed in an IC-planar drawing of $G_{\alpha}$.
\end{lemma}

\begin{proof}  
	This lemma follows directly from the construction. Each~$K_4$ must be embedded
  as a kite, and further edge crossings violate IC-planarity. 
\end{proof}

Consequently, the following corollary holds.

\begin{corollary}
  A path from a free boundary vertex cannot cross any M-graph.
\end{corollary}

Now, we show that every clause, terminal and literal gadget
has at least one primary vertex that is not \OUT.

\begin{lemma} \label{lem:inoutC-gadget}
	In every IC-planar drawing of~$G_{\alpha}$, at most two of the
	primary vertices~$u_1$,~$u_2$ and~$u_3$ of a clause gadget can be \OUT,
  and at most one of the primary vertices~$x_i^+, x_i^-$ of a terminal or a 
  literal gadget can be \OUT of the local membrane.
\end{lemma}

\begin{proof}
	Assume that~$u_1$,~$u_2$ and~$u_3$ all are \OUT. Then, the membrane 
  must cross five edges. Since the membrane has length seven, either
  one membrane edge is crossed twice or two adjacent membrane edges
  are crossed, which contradicts the IC-planarity of the drawing.
  The proof for terminal and literal gadgets works
  analogously.
\end{proof}

Next, we show that the outer membrane crosses each rope.

\begin{lemma} \label{lem:ropeXmemberane}
	In every IC-planar drawing of $G_{\alpha}$, each rope
        connected to a V-gadget is crossed by
	the outer membrane of the V-gadget.
\end{lemma}

\begin{proof}
  Suppose some rope is not crossed by the outer membrane.
  Either the outer membrane crosses at least one barrier or it crosses
  the three paths of length two that connect the V-gadget endpoint of
  the rope to its M-graph.  It cannot do the first if the size of the
  barrier is chosen to be
  \begin{displaymath}
    l \geq  \max \{k\mid\text{a variable }x\text{ occurs in }k\text{
      clauses of }\alpha\} +2 .
  \end{displaymath}
  It cannot do the second since the outer membrane, of length $2k+1$,
  would be crossed at least $k+2$ times.
\end{proof} 

The fact that a rope propagates a truth value is due to the fact
that its length is tight, as the following lemma shows.

\begin{lemma}
  In every IC-planar drawing of $G_{\alpha}$, respecting the rotation
  system, the endpoint of a rope at a C-gadget is \OUT if the endpoint
  at the vertex is \IN (its local membrane).
\end{lemma}

\begin{proof}
  M-graphs, by construction, cannot be crossed by a rope. Thus, the rope
  must cross a barrier of~$l$ edges.
  In addition, a rope is crossed by the
  outer membrane of the variable gadget. If the endpoint at the
  vertex is \IN (its local membrane), then the rope is crossed by
  $l+2$ edges. Hence, it cannot cross another membrane, since its length
  is $2l+3$.
\end{proof}

The consistency of the truth assignment of the variable is granted by the 
following lemma.

\begin{lemma}
  In every IC-planar drawing of~$G_{\alpha}$, and every variable~$x$,
  either $x_i^+$ is \OUT (and $x_i^-$ is \IN) for all $0 \leq i \leq
  k+1$
  or $x_i^+$ is \IN (and $x_i^-$ is \OUT) for all $0 \leq i \leq
  k+1$.
\end{lemma}

\begin{proof}
  If~$x_0^+$ is \OUT, then by 
  Lemma~\ref{lem:inoutC-gadget} $x_0^-$ is \IN and the local membrane 
  must cross an edge of the path of length two from~$x_0^+$ to~$x_1^-$.
  This implies that the local membrane of the first literal gadget
  cannot cross this path, and therefore must cross the paths from~$x_1^+$ to the
  M-graph. It follows by induction that all~$x_i^+$
  are \OUT and all~$x_i^-$ are \IN; see
  Figure~\ref{fi:variable-gadget}.
  If $x_0^+$ is \IN, then the outer membrane insures that $x_{k+1}^-$
  is \OUT.
  We then proceed from right to left. Now, all~$x_i^-$ are \OUT and
  all~$x_i^+$ are \IN.
\end{proof}

With these lemmas, we can finally prove the following theorem.
\newcommand{\thNpRot}{IC-planarity testing with given rotation system is \NP-complete.}
\begin{theorem}\label{th:np-rot}
  \thNpRot
\end{theorem}
\begin{proof}
  We have already stated in the proof of Theorem~\ref{th:np-hard} that
  IC-planarity is in \NP. We reduce from planar-3SAT and show
  that an expression~$\alpha$ is satisfiable if and only if the
  constructed graph~$G^*_{\alpha}$ has a IC-planar drawing.
  If~$\alpha$ is satisfiable, then we draw the V- and C-gadgets according
  to the assignment, such that the initial vertex of each rope from the
  gadget of a variable~$x$ is \IN at the C-gadget if the literal is
  assigned the value true. The resulting drawing is IC-planar by construction. 
  If~$G^*_{\alpha}$ has an IC-planar drawing, then we obtain the truth 
  assignment of~$\alpha$ from the drawing. Thus, IC-planarity with a given
  rotation system is \NP-complete.
\end{proof}

Note that the construction for the proof of Theorem.~\ref{th:np-rot} also holds 
in the variable embedding setting, since the used graphs have an almost fixed 
IC-planar embedding. From this, we can obtain an alternative \NP-completeness
proof of IC-planarity testing.

\begin{figure}[t]	
\centering
\subfigure[]{\includegraphics[scale=0.65,page=1]{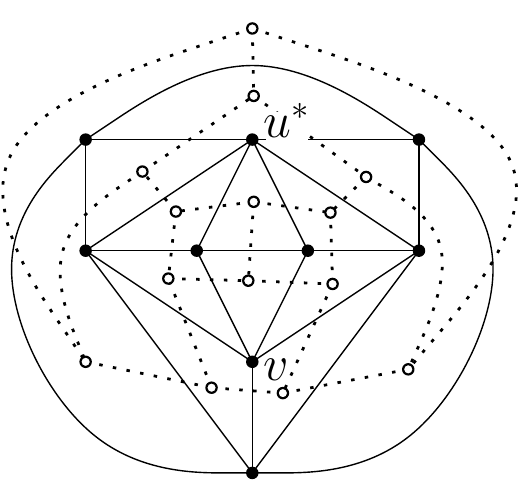}\label{fig:dual-1}}\hfil
\subfigure[]{\includegraphics[scale=0.65,page=2]{dual}\label{fig:dual-2}}\hfil
\subfigure[]{\includegraphics[scale=0.65]{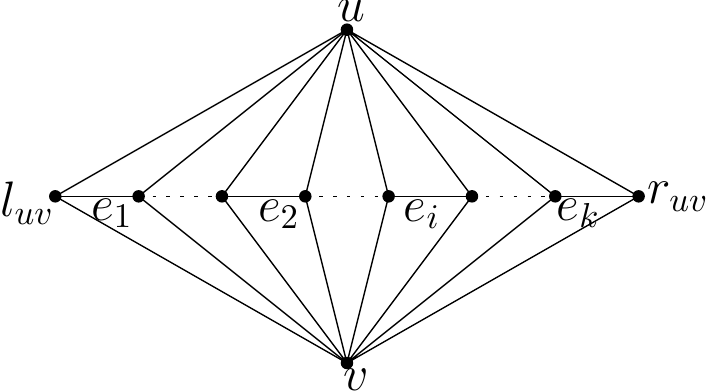}\label{fig:routing-edges}}
  \caption{(a) A triconnected graph~$T$ (solid) and its dual~$T^*$ (dotted), 
    (b) The extended graph~$T^*\cup\{u^*,v^*\}$ and the three length-3 paths
    between~$u^*$ and~$v^*$ (bold). (c) The ordered routing 
    edges~$e_1,\ldots,e_k$ lie inside the quadrangle~$(u,l_{uv},v,r_{uv})$.}
  \label{fig:dual}
\end{figure}

\subsection{Polynomial-time test for a triangulated plane graph plus a matching} 
On the positive side, we now describe an $O(n^3)$-time algorithm to test whether a graph~$G=(V,E_T \cup E_M)$
that consists of a triangulated plane graph~$T=(V,E_T)$ and a 
matching~$M=(V_M,E_M)$ with~$V_M\subseteq V, E_M\cap E_T=\emptyset$ 
admits an IC-planar drawing that preserves the embedding of~$T$. In the positive case the algorithm also computes an IC-planar drawing. 
An outline of the algorithm is as follows. 
\begin{enumerate*}[label=(\arabic{*})]
  \item \label{alg1} Check for every matching edge if there is a 
    way to draw it such that it crosses only one edge of~$T$.
  \item \label{alg2} Split~$T$ into subgraphs that form a 
    hierarchical tree structure. 
  \item \label{alg3} Traverse the 4-block tree bottom-up and solve
    a 2SAT formula for each tree node.
\end{enumerate*}

In order to check whether there is a valid placement for each matching 
edge~$(u,v)\in M$, we have to find two adjacent faces, one of which is incident
to~$u$, while the other one is incident to~$v$. To this end, we consider the 
dual~$T^*$ of~$T$ that contains a vertex for each face in~$T$ that is not
incident to a vertex~$w\in V_M\setminus\{u,v\}$, and an edge for
each edge in~$T$ that separates two faces. Further, we add two additional
vertices~$u^*$ and~$v^*$ to~$T^*$ that are connected to all faces that are
incident to~$u$ and~$v$, respectively. In the resulting graph~$T^*\cup\{u^*,v^*\}$,
we look for all paths of length~3 from~$u^*$ to~$v^*$. These paths are 
equivalent to routing~$(u,v)$ through two faces that are separated by a 
single edge. Note that no path of length~1 or~2 can exist, since $(i)$ by
construction~$u^*$ and~$v^*$ are not connected by an edge and $(ii)$ if there 
was a path of length~2 between~$u^*$ and~$v^*$, then~$u$ and~$v$ would lie on a
common face in the triangulated graph~$T$; thus, the edge~$(u,v)$ would exist
both in~$E_T$ and in~$E_M$, which is not possible since 
$E_T \cap E_M = \emptyset$. See Figure~\ref{fig:dual} for an illustration. 
If there is an edge that has no valid placement, then~$G$ is not
IC-planar and the algorithm stops. Otherwise, we save each
path that we found as a possible routing for the corresponding edge in~$M$.

Now, we make some observations on the structure of the possible routings of
an edge~$(u,v)\in M$ that we can use to get a hierarchical tree structure of
the graph~$T$. Every routing is uniquely represented by an edge
that separates a face incident to~$u$ and a face incident to~$v$ and that might
be crossed by~$(u,v)$. We call these
edges \emph{routing edges}. Let there be~$k$ routing edges for the pair~$(u,v)$.
Each of these edges forms a triangular face with~$u$. From the embedding, we
can enumerate the edges by the counterclockwise order of their corresponding
faces at~$u$. This gives an ordering~$e_1,\ldots,e_k$ of the routing edges.
Let~$e_1=(l_{uv},l'_{uv})$ and~$e_k=(r'_{uv},r_{uv})$ such that the 
edge~$(u,l_{uv})$ comes before the edge~$(u,l'_{uv})$, and the 
edge~$(u,r'_{uv})$ comes before~$(u,r_{uv})$ in the counterclockwise order
at~$u$. Then, all edges~$e_1,\ldots,e_k$ lie within the 
\emph{routing quadrangle}~$(u,l_{uv},v,r_{uv})$; see 
Figure~\ref{fig:routing-edges}. Note that there may be more complicated structures 
between the edges, but they do not interfere with the ordering. 
Denote by~$Q_{uv}=(u,l_{uv},v,r_{uv})$ the routing quadrilateral
of the matching edge~$(u,v)\in M$. We define the 
\emph{interior}~$\mathcal I_{uv}=(\mathcal V_{uv},\mathcal E_{uv})$ 
as the maximal subgraph of~$T$ such that, for every 
vertex~$w\in \mathcal V_{uv}$, each path from~$w$ to a vertex on the outer face
of~$T$ contains~$u$, $l_{uv}$, $v$, or~$r_{uv}$. 
Consequently,~$Q_{uv}\in\mathcal V_{uv}$. We will now show that two interiors
cannot overlap.

\newcommand{\lemInteriors}[1]{
  For each pair of interiors~$\mathcal I_{uv},\mathcal I_{ab}$, exactly one of 
  the following conditions holds: 
  \begin{enumerate*}[label=(\alph*)]
    \item \label{#1-a1} $\mathcal I_{uv}\cap\mathcal I_{ab}=\emptyset$ 
    \item \label{#1-a2} $\mathcal I_{uv}\subset\mathcal I_{ab}$ 
    \item \label{#1-a3} $\mathcal I_{ab}\subset\mathcal I_{uv}$ 
    \item \label{#1-a4} $\mathcal I_{uv}\cap\mathcal I_{ab}=Q_{uv}\cap Q_{ab}$. 
  \end{enumerate*}
}

\begin{lemma}\label{lem:interiors}
	\lemInteriors{main}
\end{lemma}
\begin{proof}
  Assume that neither of the conditions holds. Recall that~$Q_{uv}$ and~$Q_{ab}$
  are the boundaries of the interiors. Note that $\mathcal I_{uv}\cap\mathcal I_{ab}=\emptyset$
  corresponds to disjointness, 
  $\mathcal I_{uv}\subset\mathcal I_{ab}$ and
  $\mathcal I_{ab}\subset\mathcal I_{uv}$ correspond to inclusion, and
  $\mathcal I_{uv}\cap\mathcal I_{ab}=Q_{uv}\cap Q_{ab}$ 
  corresponds to touching in their
  boundary of the two interiors~$\mathcal I_{uv}$ and~$\mathcal I_{ab}$.
  Thus, if the conditions do not hold, the interiors must properly intersect,
  that is, without loss of generality, there is
  a vertex~$c\in Q_{uv}$ that lies in~$\mathcal I_{ab}\setminus Q_{ab}$,
  and a vertex~$d\in Q_{uv}$ that does not lie in~$\mathcal I_{ab}$. Hence, 
  the other two vertices of~$Q_{uv}$ lie in~$Q_{ab}$. Clearly,~$c$ and~$d$
  are opposite vertices in~$Q_{uv}$. By definition of IC-planar graphs, it holds 
  that~$\{a,b\}\cap\{u,v\}=\emptyset$. 
  
  First, assume that~$c=l_{uv}$. Then,~$u$ and~$v$ must lie in~$Q_{ab}$. More
  specifically, by definition of IC-planar graphs~$\{u,v\}=\{l_{ab},r_{ab}\}$.
  Without loss of generality, assume that~$u=r_{ab}$ and~$v=l_{ab}$.
  Since the edges~$(u,c)$ and~$(v,c)$ have to lie in~$\mathcal I_{ab}$, this 
  leads to the situation depicted in Figure~\ref{fi:interiors-1}. However,
  this implies that that there are only two routing edges for~$(a,b)$ with one of them
  incident to~$u$, and the other one is incident to~$v$. Thus, the routing edges
  are not valid. The case that~$c=r_{uv}$ works analogously.
  
  Second, assume that~$c=u$. Then,~$l_{uv}$ and~$r_{uv}$ must lie in~$Q_{ab}$.
  If~$l_{uv}$ and~$r_{uv}$ are adjacent on~$Q_{ab}$, say~$l_{uv}=b$ 
  and~$r_{uv}=r_{ab}$, then there is only a single routing edge for~$(u,v)$
  that is incident to~$b$ and thus not valid; see Figure~\ref{fi:interiors-2}.
  Otherwise, there are two cases. If~$\{l_{uv},r_{uv}\}=\{a,b\}$, say~$r_{uv}=a$ 
  and~$l_{uv}=b$, then there are only two routing edges for~$(u,v)$ with one
  of them incident to~$a$, and the other one incident to~$b$; see 
  Figure~\ref{fi:interiors-3}. If~$\{l_{uv},r_{uv}\}=\{l_{ab},r_{ab}\}$,
  say~$l_{uv}=l_{ab}$ and~$r_{uv}=r_{ab}$, then both routing edges of~$(u,v)$
  are incident to~$b$; see Figure~\ref{fi:interiors-4}. The case that~$c=v$
  works analogously.
  
  This proves that, if there is a proper intersection between two routing 
  quadrilaterals, than at least one of the corresponding matching edges has no
  valid routing edge. Thus, one of the conditions must hold.  
\end{proof}

 \begin{figure}[tb]
\centering
\subfigure[]{\includegraphics[scale=0.9,page=1]{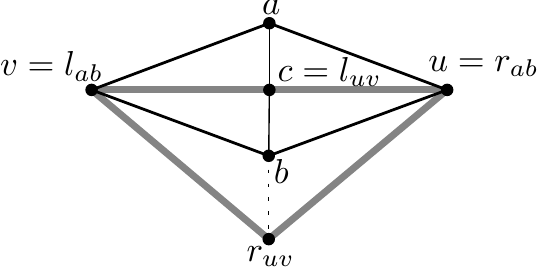}\label{fi:interiors-1}}\hfil
\subfigure[]{\includegraphics[scale=0.9,page=2]{interiors}\label{fi:interiors-2}}\\
\subfigure[]{\includegraphics[scale=0.9,page=3]{interiors}\label{fi:interiors-3}}\hfil
\subfigure[]{\includegraphics[scale=0.9,page=4]{interiors}\label{fi:interiors-4}}
    \caption{Illustration of the proof of Lemma~\ref{lem:interiors}.
      The routing quadrilateral~$Q_{ab}$ is drawn bold, and~$Q_{uv}$ is drawn dark gray and thicker. 
			(a) $l_{uv}$ lies in~$\mathcal I_{ab}\setminus Q_{ab}$. 
			(b) $l_{uv}$ and~$r_{uv}$ are adjacent in~$Q_{ab}$. 
			(c) $\{l_{uv},r_{uv}\}=\{a,b\}$. 
			(d) $\{l_{uv},r_{uv}\}=\{l_{ab},r_{ab}\}$.}
    \label{fi:interiors}
  \end{figure}

By using Lemma~\ref{lem:interiors}, we can find a hierarchical structure on the 
routing quadrilaterals. We construct a directed graph~$H=(V_H,E_H)$ with
$V_H=\{\mathcal I_{uv}\mid (u,v)\in M\}\cup \{G\}$.
For each pair~$\mathcal I_{uv},\mathcal I_{xy}$, $E_H$ contains a directed 
edge~$(\mathcal I_{uv},\mathcal I_{xy})$ if and only if 
$\mathcal V_{uv}\subset\mathcal V_{xy}$ and there is no matching edge~$(a,b)$ 
with~$\mathcal V_{uv}\subset\mathcal V_{ab}\subset\mathcal V_{xy}$.
Finally, we add an edge from each subgraph that has no outgoing edges to~$G$. 
Each vertex but~$G$ only has one outgoing edge. Obviously, this graph contains 
no (undirected) cycles. Thus,~$H$ is a tree.

We will now show how to construct a drawing of~$G$ based on~$H$ in a bottom-up
fashion. We will first look at the leaves of the graph. 
Let~$\mathcal I_{uv}$ be a vertex of~$H$ whose children are all leaves.
Let~$\mathcal I_{u_iv_i},\ldots,\mathcal I_{u_kv_k}$ 
be these leaves. Since these interiors are all leaves in~$H$, we can pick any
of their routing edges. However, the interiors may touch on their boundary, so
not every combination of routing edges can be used. Assume that a matching 
edge~$(u_i,v_i),1\le i\le k$ has more than two valid routing edges. Then, we
can always pick a \emph{middle} one, that is, a routing edge that is not 
incident to~$l_{u_iv_i}$ and~$r_{u_iv_i}$, since this edge will not interfere
with a routing edge of another matching edge. 

Now, we can create a 
2SAT formula to check whether there is a valid combination of routing edges
as follows. For the sake of clarity, we will create several redundant variables
and formulas. These can easily be removed or substituted by shorter structures
to improve the running time. For each matching edge~$(u_i,v_i),1\le i\le k$, we
create two binary variables~$l_i$ and~$r_i$, such that~$l_i$ is \true if
and only if the routing edge incident to~$l_{u_iv_i}$ is picked, and~$r_i$ is
\true if and only if the routing edge incident to~$r_{u_iv_i}$ is
picked. If~$(u_i,v_i)$ has only one routing edge, then it is obviously incident
to~$l_{u_iv_i}$ and~$r_{u_iv_i}$, so we set $l_{u_iv_i}=r_{u_iv_i}=\true$
by adding the clauses~$l_{u_iv_i}\vee\false$ 
and~$r_{u_iv_i}\vee\false$. If~$(u_i,v_i)$ has exactly two routing
edges, the picked routing edge has to be incident to either~$l_{u_iv_i}$ 
or~$r_{u_iv_i}$, so we add the clauses~$l_{u_iv_i}\vee r_{u_iv_i}$ 
and~$\neg l_{u_iv_i}\vee\neg r_{u_iv_i}$. If~$(u_i,v_i)$ has more than two 
routing edges, we can pick a middle one, so we set $l_{u_iv_i}=r_{u_iv_i}=\false$
by adding the clauses~$\neg l_{u_iv_i}\vee\false$ 
and~$\neg r_{u_iv_i}\vee\false$. Next, we have to add clauses to
forbid pairs of routing edges that can not be picked simultaneously, i.e.,
they share a common vertex. Consider a pair of matching 
edges~$(u_i,v_i),(u_j,v_j),1\le i,j\le k$. If~$r_{u_iv_i}$=$l_{u_jv_j}$, we add
the clause $\neg r_i\vee \neg l_j$. For the other three cases, we add
an analogue clause. 

Now, we use this 2SAT to decide whether the 
subgraph~$I_{uv}$ is IC-planar, and which routing edges can be used. For each
routing edge~$(a,b)$ of~$I_{uv}$, we solve the 2SAT formula given above with
additional constraints that forbid the use of routing edges incident to~$a$ 
and~$b$. To that end, add the following additional clauses: If~$l_{u_iv_i}=a$,
add the clause~$\neg l_i\vee\false$. For the other three cases, we add an 
analogue clause. If this 2SAT formula has no solution, then the 
subgraph~$\mathcal I_{uv}$ is not IC-planar. Otherwise, there is a solution 
where you pick the routing edges corresponding to the binary variables. 
To decide whether a subgraph~$I_{uv}$ whose children are not all leaves is 
IC-planar, we first compute which of their routing edges can be picked by 
recursively using the 2SAT formula above. Then, we use the 2SAT formula 
for~$I_{uv}$ to determine the valid routing edges of~$I_{uv}$. Finally, we can
decide whether~$G$ is IC-planar and, if yes, get a drawing by solving the
2SAT formula of all children of~$G$.

Hence, we give the following for the proof of the time complexity.

\newcommand{\triangtest}{
  Let $T=(V,E_T)$ be a triangulated plane graph with $n$ vertices and let 
  $M=(V,E_M)$ be a matching. There exists an $O(n^3)$-time algorithm to test if 
	$G=(V,E_T \cup E_M)$ admits an IC-planar drawing that preserves the embedding
	of~$T$. If the test is positive, the algorithm computes a feasible drawing.
}
\begin{theorem}\label{th:triang-test}
	\triangtest
\end{theorem}
\begin{proof}
We need to prove that the described algorithm runs in $O(n^3)$ time. Indeed, for each subgraph~$I_{uv}$, we 
have to run a 2SAT formula for each routing edge. Once we have determined the
valid routing edges, we do not have to look at the children anymore. 
Let~$c_{uv}$ be the number of children of~$I_{uv}$. Each of these 2SAT formula contains $2c_{uv}$ variables and up to $4c_{uv}(c_{uv}-1) $ clauses. Since every edge
of~$G$ can only be a routing edge for exactly one matching edge, we have to 
solve at most~$n$ 2SAT formulas. The tree~$H$ consists of at most~$n/2+1$
vertices (one for each matching edge), so a very conservative estimation is
that we have to solve~$O(n)$ 2SAT formulas with~$O(n)$ variables and~$O(n^2)$
clauses each. Aspvall {\em et al.}~\cite{apt-ltatt-79}  
showed how to solve 2SAT in
time linear in the number of clauses. We can use the linear-time algorithm of
Section~\ref{ic:sec:drawing} to draw the IC-planar graph corresponding to the
IC-planar embedding by picking the routing edges corresponding to the binary
variables. Thus, our algorithm runs in~$O(n^3)$ time.
\end{proof}

\section{Open Problems}\label{se:conclusions}
The research presented in this paper suggests interesting open problems.
 
\begin{description}
\item[Problem 1.] We have shown that every IC-planar graph has a straight-line drawing in quadratic 
area, although the angle formed by any two crossing edges can be small. 
Conversely, straight-line \Rac drawings of IC-planar graphs may require 
exponential area. From an application perspective, it is interesting to design algorithms that compute a straight-line drawing of IC-planar graphs in polynomial area and good crossing resolution.

\item[Problem 2.] Also, although IC-planar graphs are both 1-planar and straight-line \Rac drawable graphs, a characterization of the intersection between these two classes is still missing. In particular, studying whether \emph{NIC-planar graphs} (see Zhang~\cite{z-dcmgp-AMS14}), which lie between IC-planar graphs and 1-planar graphs, are also \Rac graphs may lead to new insights on this problem.

\item[Problem 3.] We proved that recognizing IC-planar graphs is NP-hard. Is it possible to design fixed-parameter tractable (FPT) algorithms for this problem, which improve the time complexity of those described by Bannister \emph{et al.}~\cite{DBLP:conf/wads/BannisterCE13} for 1-planar graphs, with respect to different parameters (vertex cover number, tree-depth, cyclomatic number)? Are there other parameters that can be conveniently exploited for designing FPT testing algorithms for IC-planar graphs?  
\end{description}

\section*{Acknowledgments}
We thank the anonymous reviewers of this work for their useful comments and suggestions.
We also thank Michael A. Bekos and Michael Kaufmann for suggesting a simpler counterexample for the area requirement of IC-plane straight-line RAC drawings.

\clearpage

{\small \bibliography{abbrv,icplanar}}
\bibliographystyle{abbrv}

\end{document}